\documentclass[keeplastbox,conference]{style/IEEEtran}
\IEEEoverridecommandlockouts
\usepackage{amssymb}
\usepackage{amsthm}
\usepackage{amsmath}
\usepackage{color}
\usepackage{enumitem}
\usepackage{arydshln}
\usepackage{subcaption}
\usepackage{graphicx}
\usepackage{comment}
\usepackage{tikz}
\usepackage{pgfplots}
\pgfplotsset{every tick label/.append style={font=\small}}

\usepackage{algorithm}
\usepackage{algorithmic}

\usepackage{flushend}
\usepackage{url}

\usepackage[colorinlistoftodos]{todonotes}
\newcommand\eat[1]{}

\newcommand\del[1]{}
\newcommand\submit[1]{}
\newcommand\techrep[1]{#1}
\newcommand\revision[1]{#1}

\newcommand\new[1]{#1}

\newtheorem{theorem}{Theorem} 
\newtheorem{definition}{Definition}
\newtheorem{example}{Example}
\newtheorem{problem}{Problem}
\newcommand{\stitle}[1]{\vspace{1ex}\noindent{\bf #1}}

\newcommand{\naive}{na\"{i}ve}

\usepackage{xspace}

\newcommand{\patterncombiner}{{\sc pattern-combiner\xspace}}
\newcommand{\patternbreaker}{{\sc pattern-breaker\xspace}}
\newcommand{\deepdiver}{{\sc deepdiver\xspace}}

\newcommand{\greedy}{{\sc greedy}\xspace}
\newcommand{\apriori}{{\sc apriori}\xspace}

\def\BibTeX{{\rm B\kern-.05em{\sc i\kern-.025em b}\kern-.08em
    T\kern-.1667em\lower.7ex\hbox{E}\kern-.125emX}}

\begin{document}
\title{Assessing and Remedying Coverage\\ for a Given Dataset}

\author{\IEEEauthorblockN{Abolfazl Asudeh, Zhongjun Jin, H. V. Jagadish}
\IEEEauthorblockA{University of Michigan}
\textit{$\{$asudeh, markjin, jag$\}$@umich.edu}
}

\maketitle
\begin{abstract}
Data analysis impacts virtually every aspect of our society today.  
Often, this analysis is performed on an existing dataset, possibly collected through a process that the data scientists had limited control over.
The existing data analyzed may not include the complete universe, but it is expected to cover the diversity of items in the universe.
Lack of adequate coverage in the dataset can result in undesirable outcomes such as biased decisions and algorithmic racism, as well as creating vulnerabilities such as opening up room for adversarial attacks.

In this paper, we assess the coverage of a given dataset over multiple categorical attributes.
We first provide efficient techniques for traversing the combinatorial explosion of value combinations to {\em identify} any regions of attribute space not adequately covered by the data. 
Then, we determine the least amount of additional data that must be {\em obtained} to resolve this lack of adequate coverage.
We confirm the value of our proposal through both theoretical analyses and comprehensive experiments on real data.
\end{abstract}

\section{Introduction}
\label{sec:intro}
In the current age of data science, it is commonplace to have a learning algorithm trained based on some dataset.
\new{
This dataset could be collected prospectively, such as through a survey or a scientific experiment. In such a case, a data scientist may be able to specify requirements such as representation and coverage.
However, more often than not, analyses are done with data that has been acquired independently, possibly through a process on which the data scientist has limited, or no, control. This is often called ``found data'' in the data science context.}
It is generally understood that the training dataset must be representative of the distribution from which the actual test/production data will be drawn.  More recently, it has been recognized that it is not enough for the training data to be representative: it must include enough examples from less popular ``categories'', if these categories are to be handled well by the trained system.  Perhaps the best known story underlining the importance of this inclusion is the case of the ``google gorilla'' \cite{google-gorilla}.  An early image recognition algorithm released by Google had not been trained on enough dark-skinned faces.  When presented with an image of a dark African American, the algorithm labeled her as a ``gorilla''.  While Google very quickly patched the software as soon as the story broke, the question is what it could have done beforehand to avoid such a mistake in the first place.

The Google incident is not unique: there have been many other such incidents.  For example, Nikon introduced a camera feature to detect whether humans in the image have their eyes open -- to help avoid the all-too-common situation of the camera-subject blinking when the flash goes off resulting in an image with eyes closed.  Paradoxically for a Japanese company, their training data did not include enough East Asians, so that the software classified many (naturally narrow) Asian eyes as closed even when they were open \cite{closed-eyes}.
Similarly, HP webcams were not able to detect black faces~\cite{hp1} due to inadequate coverage in the training data~\cite{hp2}.

The problem becomes critical when it comes to {\em data-driven algorithmic decision making}.
For example, judges, probation and parole officers are increasingly using algorithms to assess a criminal defendant's likelihood to re-offend~\cite{propublica}.
Consider a tool designed to help the judges in sentencing criminals by predicting how likely an individual is to re-offend. 
Such a tool can provide insightful signals for the judge and have the potential to make society safer.
On the other hand, a wrong signal can have devastating effects on individuals' lives.
So it is important to make sure that the tool is trained on data that includes adequate representation of individuals similar to each criminal that will be scored by it.
In \S~\ref{subsec:exp-validation}, we study a real dataset of criminals used for building such a tool, published by Propublica~\cite{propublica}.
We shall show how 
inadequate representation might result, for example, in predicting every widowed Hispanic female as highly likely to re-offend.

While Google's resolution to the gorilla incident was to {\em ``ban gorillas''}~\cite{google-gorilla-resolution}, a better solution is to ensure that the training data has enough entries in each category.
Referring to the issue as ``disparate predictive accuracy'', \cite{chen2018my} also highlights that the problem often is due to the insufficient or skewed sample sizes.
If the only category of interest were race, as in (most of) the examples above, there are only a handful of categories and this problem is easy.
However, in general, objects can have tens of attributes of interest, all of which could potentially be used to categorize the objects.  
For example, survey scientists use multiple demographical variables to characterize respondents, including race, sex, age, economic status, and geographic location.
Whatever be the mode of data collection for the analysis task at hand, we must ensure that there are enough entries in the dataset for each object category.  
Drawing inspiration from the literature on diversity \cite{diversity-jag}, we refer to this concept as {\em coverage}.

Note that the mentioned examples, including the Google incident, are surely not sampling cases where the data scientists poorly chose the samples from a large database. Rather, they somehow collected, or acquired, a dataset, and then failed to realize the lack of coverage for dark-skinned faces.

Lack of coverage in a dataset also opens up the room for adversarial attacks~\cite{biggio2013evasion}. The goal in an adversarial attack is to generate examples that are misclassified by a trained model.
Poorly covered regions in the training data provide the adversary with opportunities to create such examples.
%For instance, consider the gorilla incident again. Knowing that black people are under-represented in the dataset gives the adversary the information that the models trained using this dataset are not well-trained for this category. The adversary can use this information to generate examples that are misclassified by the model.

Our goal in this paper is two-fold.
First, we would like to help the dataset users \new{to be able to assess the coverage, as a characterization, of a given dataset, in order to understand such vulnerabilities.}
For example, we propose to use information about lack of coverage as a widget in the nutritional label~\cite{yang2018nutritional} of a dataset.
\new{
Once the lack of coverage has been identified, next we would like to help data owners improve coverage by identifying the smallest number of additional data points needed to hit all the ``large uncovered spaces''.
}

Given multiple attributes, each with multiple possible values, we have a combinatorial number of possible {\em patterns}, as we call  combinations of values for some or all attributes.  Depending on the size and skew in the dataset, the coverage of the patterns will vary.  Given a dataset, our first problem is to efficiently identify patterns that do not have sufficient coverage (the learned model may perform poorly in portions of the attribute space corresponding to these patterns of attribute values).  It is straightforward to do this using space and time proportional to the total number of possible patterns.  Often, the number of patterns with insufficient coverage may be far fewer.  In this paper, we develop techniques, inspired from set enumeration~\cite{setenum} and association rule mining ({\em apriori})~\cite{apriori}, to make this determination efficient.
\new{We shall further discuss this and the related work in \S~\ref{sec:related}.}

\eat{
The uncovered patterns in a dataset show its vulnerabilities.
As the dataset users, we seek to avoid using the datasets with vulnerabilities.
On the other hand, 
in a dataset with many attributes (each with multiple values),
even after limiting the consideration to a subset of {\em ``attributes of interest''},
the space of possible combinations may become too large
that it is unlikely to have sufficient coverage for each, even in a dataset with large number of diverse items.

Since there is a combinatorial number of patterns,
many of these patterns may not be covered, even with very large datasets. 
In general, uncovered patterns with many attributes specify small uncovered regions in the data and are less harmful than uncovered patterns with only a few attributes.
For example, it may be less of a problem if we have insufficient coverage of black men non-smokers over age of 60 who live in Texas than if we have insufficient coverage of black men or men over 60.   
Therefore, we may seek to make sure that we have adequate coverage for at least any pattern of $\ell$ attributes, where we call $\ell$ the {\em maximum coverage level}.
Checking this is easy, given the set of uncovered patterns, as we can use them to determine the maximum coverage level.  In fact, determining the latter does not even require that we identify all patterns with insufficient coverage: we can limit the exploration to the level $\ell$. 
}

A more interesting question for the dataset owners is what they can do about lack of coverage.
Given a list of patterns with insufficient coverage, they may try to fix these, for example by acquiring additional data.  In the ideal case, they will be able to acquire enough additional data to get sufficient coverage for all patterns.  However, acquiring data has costs, for data collection, integration, transformation, storage, etc.  Given the combinatorial number of patterns, it may just not be feasible to cover all of them in practice.  Therefore, we may seek to make sure that we have adequate coverage for at least any pattern of $\ell$ attributes, where we call $\ell$ the {\em maximum covered level}.
Alternatively, we could identify important pattern combinations by means of a {\em value count}, indicating how many combinations of attribute values match that pattern.
Hence, our goal becomes to determine the patterns for the minimum number of items we must add to the dataset to reach a desired maximum covered level or to cover all patterns with at least a specified minimum value count.  Since a single item could contribute to the coverage of multiple patterns, we shall show that this problem translates to a hitting set~\cite{vazirani2013approximation} instance. Given the combinatorial number of possible value combinations, the direct implementation of hitting set techniques can be very expensive.
We present an approximate solution technique that can cheaply provide good results.
%how a transformation to the frequent item-set problem can help in efficiently solving the problem.

We note that not all combinations of attribute values are of interest.  Some may be extremely unlikely, or even infeasible.  For example, we may find few people with attribute {\sf age} as ``teen'' and attribute {\sf education} as ``graduate degree''.
A human expert, with sufficient domain knowledge, is required to be in the loop for (i) identifying the attributes of interest, over which coverage is studied, (ii) setting up a {\em validation oracle} that identifies the value combinations that are not realistic, and (iii) identifying the uncovered patterns and the granularity of patterns that should get resolved during the coverage enhancement.

\techrep{
\noindent{\bf Summary of contributions.} 
In summary, our contributions in this paper are as follows:
\begin{itemize}[leftmargin=*]
\itemsep0em 
\item We formalize the novel notion of {\em maximal uncovered patterns (MUP)}, to show the lack of coverage with regard to multiple attributes. We define (i) the MUP Identification problem, as well as (ii) Coverage Enhancement problem for resolving the lack of coverage. We prove that no polynomial-time algorithm can exists for (i) and that (ii) is NP-hard.
\item
Introducing the pattern graph for modeling the space of possible patterns, we provide three algorithms \patternbreaker, \patterncombiner, and \deepdiver\ for efficiently discovering the MUPS.
\item For dataset owners, we formulate the problem of additional data collection, connect the problem to hitting set, and propose an efficient implementation of the greedy approximation algorithm for the problem.
\item We use empirical evaluation on real datasets to validate our proposal, and to demonstrate the efficiency of the proposed techniques. Besides the performance evaluations, we investigate the lack of coverage in a dataset of criminals' records and discuss how a tool built using it may generate wrong signals for sentencing criminals. \new{We show that a classifier with an acceptable performance on a random test set, may have a bad performance over the minority groups. We also show that remedying the lack of coverage improves the performance of the model for the minorities.}
\end{itemize}
}% end techrep

\vspace{-1mm}
\section{Preliminaries} \label{sec:pre}

%\stitle{Data Model:}
%In this paper, we study the lack of coverage over multiple categorical attributes.
We consider a dataset $\mathcal{D}$ with $d$ low-dimensional categorical attributes, $\mathcal{A} = \{A_1$, $A_2$, ..., $A_d\}$.
Where attributes are continuous valued or of high cardinalities, 
we consider using techniques such as
(a) bucketization: putting similar values into the same bucket, or (b) considering the hierarchy of attributes in the data cube for reducing the cardinalities.
Each tuple $t \in \mathcal{D}$ is a vector with the value of $A_i$ being $t[i]$ for all $i=1,...d$.
In addition, the dataset also contains the ``label attributes'' $Y = \{ y_1,\cdots,y_{d'}\}$ that contain the target values.
The label attributes are not considered for the coverage problem.
In practice, a user may be interested in studying the coverage over a subset of ``attributes of interest''. In such cases, the problem is limited to those attributes.
For instance, in a dataset of criminals, attributes such as {\tt \small sex}, {\tt \small race}, and {\tt \small age} can be attributes of interest while the label attribute shows whether or not the criminal has re-offended.
In the rest of the paper,
we assume $A_1$ to $A_d$ are the attributes of interest and simply name them as the set of attributes.
The cardinality of an attribute $A_i$ is $c_i$.
Hence, the total number of value combinations is $\Pi_{k=1}^d c_k$.
For a subset of attributes $\mathcal{A}_i\subseteq \mathcal{A}$, we use the notation $c_{\mathcal{A}_i} = \Pi_{\forall A_j\in\mathcal{A}_i} c_j$ to show the number of value combinations for $\mathcal{A}_i$.

\begin{definition}[Pattern]
A pattern $P$ is a vector of size $d$, in which $P[i]$ is either $X$ (meaning that its value is unspecified) or is a value of attribute $A_i$. We name the elements with value $X$ as non-deterministic and the others as deterministic.
\end{definition}

\noindent An item $t$  {\em matches} a pattern $P$ (written as $M(t,P)= \top$), if for all $i$ for which $P[i]$ is deterministic, $t[i]$ is equal to $P[i]$. Formally:
\begin{align}
M(t,P)=\begin{cases}
\top, & \forall i\in [1,d]:P[i]=X \mbox{ or } P[i]=t[i] \\
\bot, & \mbox{otherwise}
\end{cases}
\end{align}

For example, consider the pattern $P$ = $X1X0$ on four binary attributes $A_1$ to $A_4$.
It describes the value combinations that have the value 1 on $A_2$ and 0 on $A_4$.
Hence, for example, $t_1=[1,1,0,0]$ and $t_2=[0,1,1,0]$ match $P$, as their values on all deterministic elements of $P$ (i.e., $A_2$ and $A_4$) match the ones of $P$.
On the other hand, $t_3=[1,0,1,0]$ does not match the pattern $P$. That is because $P[2]=1$ while $t_3[2]=0$.

Using the patterns to describe the space of value combinations, we now define the coverage notion as follows:
\begin{definition}[Coverage]\label{def:coverage}
Given a dataset $\mathcal{D}$ over $d$ attributes with cardinalities $c=\{c_1\cdots c_d\}$, and a Pattern $P$ based on $c$ and $d$, the coverage of $P$ is the number of items in $\mathcal{D}$ that match $P$. Formally:
$cov(P,\mathcal{D}) = | \{ t\in\mathcal{D}~|~ M(t,P)=\top \}|$.
\end{definition}
\noindent When $\mathcal{D}$ is known, we can simplify $cov(P,\mathcal{D})$ with $cov(P)$.

We would like a high enough coverage for each pattern, to make sure it is adequately represented.  
How high is enough is expected as an input to our problem, and is expected to be determined through statistical analyses.
There is a long tradition of computing the ``power'' of an experiment design, to determine the subject pool size (corresponding to coverage) required to obtain statistically meaningful results.
\revision{
Borrowing the concept from statistics and central limit theorem, the rule of thumb suggests the number of representatives to be around 30. For example, Sudman~\cite{sudman1976applied} suggests that for each ``minor subgroup'' a minimum of 20 to 50 samples is necessary.
This is what we also observed in our experiments (\S~\ref{subsec:exp-validation}).
}
Using such, or other, techniques, we will assume that a 
{\em Coverage threshold}, $\tau$, has been established for each pattern.

\begin{definition}[Covered/Uncovered Pattern]
A pattern $P$ is said to be {\em covered} in a dataset $\mathcal{D}$ if its coverage is greater than or equal to the specified coverage threshold:
$cov(P,\mathcal{D}) \geq \tau$.
Otherwise, the pattern $P$ is said to be {\em uncovered}.
\end{definition}

Each pattern describes a region in the space of value combinations, constrained by its deterministic elements.
We define the level of each pattern $P$, shown as $\ell(P)$, as the number of deterministic elements in it.
Patterns with fewer deterministic elements (smaller level) are more general.
For example, consider two patterns $P_1=$1XXX and $P_2=$10X1 on four binary attributes $A_1$ to $A_4$.
$\ell(P_1)=1$ and $\ell(P_2)=3$.
While only the value combinations 1001 and 1011 match $P_2$, any value combination with value 1 on $A_1$ matches $P_1$.

The set of value combinations that match a pattern $P$ may be a subset of the ones that match a more general pattern $P'$. We say that $P$ is {\em dominated by}  $P'$ (or $P'$ dominates $P$).
For example, the pattern $P_2=$10X1 is dominated by the pattern $P_1=$1XXX.

\begin{definition}[Parent/Child Pattern]
A pattern $P_1$ is a parent of a pattern $P_2$ if $P_1$ can be obtained by replacing one of the deterministic elements in $P_2$ (say $P_2[i]$) with $X$. We can equivalently say that $P_2$ is a child of pattern $P_1$.
\end{definition}

In general, patterns can each have multiple parents and multiple children.  A pattern with all elements being non-deterministic has no parent and a pattern with all elements being deterministic has no child. 

If a pattern is uncovered, all of its children, and their children, recursively, must also be uncovered. When identifying uncovered patterns, it is redundant to list all these dominated uncovered patterns: doing so just makes the output much larger, and much harder for a human to digest and address.
Therefore, 
our goal is to identify the uncovered patterns that are not dominated by more general uncovered ones.
%In the following, we provide the terms and definitions that lead to the formal problem formulation.

\begin{definition}[Maximal Uncovered Pattern (MUP)]\label{def:mup}
Given a threshold $\tau$, a pattern $P$ is maximal uncovered, if $cov(P)< \tau$, while for any pattern $P'$ parent of $P$, $cov(P')\geq \tau$.
\end{definition}

\noindent With these definitions, we formally state our first problem as:

\begin{problem}[MUP Identification Problem]\label{problem1}
Given a dataset $\mathcal{D}$ defined over $d$ attributes with cardinalities $c$, as well as the coverage threshold $\tau$, find all maximal uncovered patterns $\mathcal{M}$.
\end{problem}

%After defining Problem~\ref{problem1}, next in Theorem~\ref{th:1} we show that no polynomial time algorithm can exist for addressing it.
%\jin{we probably need one sentence in between our problem definition and the following theorem}\abol{is it good now? please feel free to adjust the connection as you wish.}
While there usually are far fewer MUPs than uncovered patterns, the worst case remains bad, as we show next.

\begin{theorem}\label{th:1}
No Polynomial time algorithm can guarantee the enumeration of the set of maximal uncovered patterns.
\end{theorem}

The proof is by construction of an example with an exponential number of MUPs.  \submit{Details in the technical report~\cite{techreport}.}
\techrep{
\begin{proof}
We prove the theorem by construction.
Consider a dataset $\mathcal{D}$ (shown in the following) with $n$ items and $d=n$ binary attributes in which only the values of the elements on the diagonal are one and the rest are zero. That is, $\forall i\in [1,n]:~t_i[i]=1$ and $\forall j\neq i:~t_i[j]=0$.
	\begin{center}
    \begin{small}
	\begin{tabular}{|c|c|c|c|c|}
		\hline
		&$A_1$&$A_2$&$\cdots$&$A_n$ \\ \hline
        $t_1$&1&0&$\cdots$&0 \\ \hline
        $t_2$&0&1&$\cdots$&0 \\ \hline
        $\vdots$&\vdots&\vdots&$\ddots$&\vdots \\ \hline
        $t_n$&0&0&$\cdots$&1 \\ \hline
	\end{tabular}
    \end{small}
	\end{center}
Let the threshold $\tau$ be $\frac{n}{2}+1$.
  
First, any pattern with $(n-1)$ non-deterministic element and one deterministic element with value $1$ is a MUP.
That is because $t_i$ is the only tuple matching such a pattern where its deterministic value ($1$) appears at $i$-th element.
Since these patterns have only one non-deterministic element, their parent is the pattern $XX\cdots X$ in which there is no deterministic element. All the items in the dataset match this pattern and, hence, its coverage is $n$.
As a result, all such patterns with only one deterministic element $1$ are MUPs. The number of such patterns in $n$. Also, since any pattern with more than one deterministic element where one of them is $1$ is dominated by these MUPs, there cannot exist such a MUP. It means all the deterministic values of the remaining MUPs should be $0$.

Next, consider a pattern $P$ with $m\leq n$ deterministic elements with values $0$. Let $I$ be the indices of the deterministic elements. For all values $i\in I$, $t_i$ does not match $P$; simply because $t_i[i]=1$ while $P[i]$ is $0$. Note that all other tuples in $\mathcal{D}$, $t_i$ for $i \notin I$, match $P$. Thus, the coverage of such a pattern is $n-m$.
Now, let $m$ be $\frac{n}{2}$. The coverage of any such pattern is $\frac{n}{2}<\tau$. Similarly, the coverage of any of its parents of such nodes is $\frac{n}{2}+1=\tau$ (because it is a pattern with m - 1 deterministic elements with value 0).
As a result, every pattern $P$ with $\frac{n}{2}$ deterministic elements, all being $0$, is a MUP.
The number of such patterns are ${n \choose n/2}$.
Therefore, the total number of MUPs in $\mathcal{D}$ is:
$$|\mathcal{M}| = n + {n \choose n/2} > 2^n$$
As a result, any algorithm enumerating over them is exponential. 
\end{proof}
}

Not all MUPs are problematic.  For example, if some combination of attribute values is known to be infeasible, the corresponding pattern will necessarily be uncovered.  A domain expert can examine a list of MUPs and identify the ones that can safely be ignored.  The remaining are considered {\em material}.

%Ideally, we would like to have no material uncovered pattern in a dataset. However, this is usually not feasible in practice, even for the datasets with large number of items, because of the combinatorial number of patterns.
In many situations, large uncovered regions in the dataset are more harmful than narrow uncovered regions.
%\jin{The statement in previous two sentences are very important, but I don't understand why they are here. There are no connections before or after whatsoever.}\abol{Agreed. I moved it here, and adjusted the structure/text accordingly. is it good now?}
Following this observation, for a dataset $\mathcal{D}$, we define the {\em maximum covered level} as the maximum level up until which there is enough coverage in the dataset. Formally:
\begin{definition}[Maximum Covered Level]
Let $\mathcal{M}$ be the material MUPs for a dataset $\mathcal{D}$.
Then, the maximum covered level of $\mathcal{D}$ is the maximum level $\lambda$ such that $\forall P\in \mathcal{M}, \ell(P)>\lambda$.
\end{definition}

In light of the above, we would like to have as large a maximum covered level as possible for a dataset.  %If we are not satisfied with the maximum coverage level of a dataset, we can attempt to obtain additional data to raise this level.  However, data collection usually has associated cost, which can sometimes be substantial.  Our objective is to find the minimum data collection that assures the desired maximum coverage level. 
%\jin{we need one sentence explaining the motivation assuring an input maximum coverage level and minimizing the data collection.} \abol{is it good now?}

\begin{problem}[Coverage Enhancement Problem]\label{problem2}
Given a dataset $\mathcal{D}$, its set of material MUPs $\mathcal{M}_\mathcal{D}$, and a positive integer number $\lambda$, 
determine the minimum set of additional tuples to collect such that, after the data collection, the maximum covered level of $\mathcal{D}$ is at least $\lambda$.
\end{problem}

We can consider variants of the coverage enhancement problem where we seek to attain some other coverage property rather than satisfy a maximum coverage level.
For example, instead of the level of a pattern $P$, one could consider the number of value combination matching it.  
\begin{definition}[Value Count]\label{def:vc}
Let $\mathcal{A}_P$ be the set of corresponding attributes for non-deterministic elements of  a pattern $P$. 
The value count of $P$ is the number of value combinations matching $P$.  That is, $c_{\mathcal{A}_P} = \Pi_{\forall A_j\in\mathcal{A}_P} c_j$.
\end{definition}
For example, consider the pattern $P=X1X0$ over binary attributes $\mathcal{A}=\{A_1,\cdots,A_4\}$. $\mathcal{A}_P=\{A_1,A_3\}$. Hence, the number of value combinations matching $P$ is $c_{\mathcal{A}_P} = 2\times2 = 4$.
The coverage enhancement problem can be modified to require that every pattern $P$ in $\mathcal{D}$ be covered if the value count of $P$ is at least $v$.
As further explained in \S~\ref{sec:4}, the proposed solution can easily be extended for such alternative measures.

Next, in Theorem~\ref{th:2}, we study the complexity of the coverage enhancement problem.

\begin{theorem}\label{th:2}
The Coverage Enhancement Problem is NP-hard.
\end{theorem}
\submit{
We prove the theorem using an interesting polynomial-time reduction from the {\em vertex cover} (VC) problem~\cite{vazirani2013approximation}.
Please find the details in the technical report~\cite{techreport}.
}
\techrep{
\begin{proof}
The decision version of the Coverage Enhancement (CE) Problem
is as follows:
given a dataset $\mathcal{D}$, its set of material MUPs $\mathcal{M}_\mathcal{D}$, a positive integer number $\lambda$, and a decision value $m$, is there a set of $m$ additional tuples to collect such that after the data collection the maximum coverage level of $\mathcal{D}$ is at least $\lambda$.
We use the polynomial-time reduction from the {\em vertex cover} (VC) problem~\cite{vazirani2013approximation}.
The decision version VC, given an unweighted undirected graph $G(V,E)$ and a decision variable $m$, decides if there is a subset of $m$ vertices such that each edge is incident to at least one vertex of the set.

\begin{figure}[t!]
\centering
\subcaptionbox{\footnotesize\label{fig:proofG} A sample graph for the vertex cover problem.}{\includegraphics[width=0.16 \textwidth]{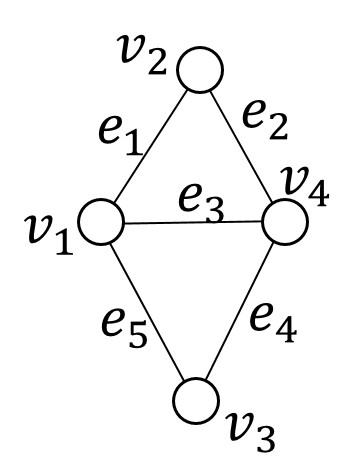}}
\hfill 
\subcaptionbox{\footnotesize\label{fig:proofD} The constructed dataset (and its MUPs) for graph of Figure~\ref{fig:proofG}.}
{	
\begin{scriptsize}
\begin{tabular}{|l||c|c|c|c|c|}
	\hline
	      &$A_1$&$A_2$&$A_3$&$A_4$&$A_5$\\ \hline
    $t_1$ &  1  &  0  &  1  &  0  &  1\\ \hline
    $t_2$ &  1  &  1  &  0  &  0  &  0\\ \hline
    $t_3$ &  0  &  0  &  0  &  1  &  1\\ \hline
    $t_4$ &  0  &  1  &  1  &  1  &  0\\ \hline
    $t_5$ &  0  &  0  &  0  &  0  &  0\\ \hline
    $t_6$ &  0  &  0  &  0  &  0  &  0\\ \hline
    $t_7$ &  0  &  0  &  0  &  0  &  0\\ \hline \hline
    $P_1$ &  1  &  X  &  X  &  X  &  X\\ \hline
    $P_2$ &  X  &  1  &  X  &  X  &  X\\ \hline
    $P_3$ &  X  &  X  &  1  &  X  &  X\\ \hline
    $P_4$ &  X  &  X  &  X  &  1  &  X\\ \hline
    $P_5$ &  X  &  X  &  X  &  X  &  1\\ \hline
	\end{tabular}
\end{scriptsize}
	}
\vspace{-3mm}\caption{\footnotesize 
\label{fig:proof}
Illustration of the reduction from vertex cover to coverage enhancement problem.
}
\vspace{-6mm}
\end{figure}

The reduction is as follows:
Given a graph $G(V,E)$ for VC, construct the dataset $\mathcal{D}$ with $n=|V|+3$ items and $d=|E|$ attributes
(Figure~\ref{fig:proofD} shows the reduction for the sample graph of Figure~\ref{fig:proofG}).
For every edge $e_j\in E$ add the attributes $A_j$ to $\mathcal{D}$.
For every vertex $v_i\in V$, add the item $t_i$ to the dataset.
For every edge $e_j$ that is connected to $v_i$, set $t_i[j]$ to $1$; set the rest of attribute values of $t_i$ as $0$.
Add three items $t_{|v|+1}$, $t_{|v|+2}$, and $t_{|v|+3}$ all with attribute values being $0$.
Set the coverage threshold to $\tau=3$ and the maximum coverage level to $\lambda=1$.
In this dataset, any pattern with only one deterministic element with value $1$ is a MUP. That is because, (i) for any such pattern $P_j$ having $1$ in its $j$-th element, the corresponding items with the two vertices in $G$ that are connected to $e_j$ are the only items matching it; also, (ii) the only parent of such patterns, $XX\cdots X$ has the coverage above the threshold, as all the items match it.
Consequently, any pattern with a more than one deterministic element with at least one $1$ is dominated by a MUP and, therefore, is not a MUP.
Moreover, any pattern with some deterministic elements with values $0$ is covered by $t_{|v|+1}$, $t_{|v|+2}$, and $t_{|v|+3}$ and, thus, is not a MUP. As a result, the patterns $P_j$ (with only one deterministic element with value $1$ in the $j$-th element) are the collection of MUPs for $\mathcal{D}$.
There totally are $|E|$ of those patterns (pattern $P_j$ is associated with $e_j$).
Having the maximum coverage level as $\lambda=1$, CE needs to cover all of these patterns; hence, if there exists a subset of size $m$ of items covering these patterns, their equivalent vertices are the subset of $m$ vertices such that each edge is incident to at least one vertex of the subset.

\noindent Given this reduction, there is no polynomial-time algorithm for the CE Problem, unless $P=NP$.
\end{proof}
}%end techrep

\section{MUP Identification}\label{sec:3}
In this section, we study Problem~\ref{problem1}, MUP identification, and propose efficient search and pruning strategies for it.

\subsection{Na\"{i}ve}\label{subsec:mup-naive}
A single pass over the dataset can suffice, with one counter for each pattern.  With one pass, we obtain the count for each pattern, and can determine which patterns are uncovered.  We can then compare each pair of uncovered patterns, \{$P_i$,$P_j$\}.  If $P_i$ dominates $P_j$, then the latter is not maximal, and can be removed from the list of uncovered patterns discovered.  After all pairs of uncovered patterns have been processed, and the ones not maximal eliminated, then the remaining uncovered patterns are the desired maximal uncovered patterns.

Suppose there are $d$ attributes\del{, each of which takes $c$ possible values}.
Each element of a pattern can either be non-deterministic, or a value from the corresponding attribute.
%In a pattern, the "value" assigned to an attribute can be $X$, in addition to the $c$ actual values.  
\new{As such, there are $c_i+1$ choices for each attribute $A_i$, resulting in a total of $c^+_\mathcal{A} = \Pi_{k=1}^d(c_i+1)$ patterns.
We need one counter for each pattern, or a total space of $O(c^+_\mathcal{A})$.  The time taken to find all uncovered patterns is $O(n\times c^+_\mathcal{A})$, where there are $n$ tuples in the dataset.  
Let the total number of uncovered patterns found be $u$.  Then an additional $O(u^2)$ time is required to find the maximal uncovered patterns from among these.  
Thus, the total time required is $O(n~c^+_\mathcal{A} ~+~ u^2)$.
While the additional time due to the second term will usually be smaller than the first term, we note that $u$ could be as large as $c^+_\mathcal{A}$, and is usually much larger than the number of maximal uncovered patterns.
}
As a toy example, consider Example~\ref{ex:1}.
\begin{example}\label{ex:1}
Consider a dataset $\mathcal{D}$ with binary attributes $A_1$, $A_2$, and $A_3$, containing the tuples $t_1: 010$, $t_2: 001$, $t_3: 000$, $t_4: 011$, and $t_5: 001$.
Let the coverage threshold be $\tau = 1$.
\end{example}
The dataset in Example~\ref{ex:1} has one MUP 1XX.
In addition to the MUP, the other 8 uncovered patterns (dominated by the only MUP) are 1X0, 1X1, 10X, 11X, 100, 101, 110, and 111. 

%\jag{Brief example here to explain this, possibly based on Fig. 1}

\subsection{Pattern Graph}\label{subsec:patterngraph}

In the na\"{i}ve algorithm, we computed all uncovered patterns, only to eliminate those that were not maximal.  It would appear that we could do less work if we could exploit relationships between patterns.  Specifically, patterns have parent/child relationships, as discussed in \S~\ref{sec:pre}. We can represent relationships between patterns by means of a pattern graph, and use this data structure to find better algorithms.

\begin{definition}[Pattern Graph]
Let $\mathcal{P}$ be the set of all possible patterns defined over $d$ attributes with cardinalities $c$.
Pattern graph of $\mathcal{P}$ is the graph $G(V,E)$ where $V=\mathcal{P}$.
There is an edge between every pair of nodes $P$ and $P'$ that have a parent-child relationship.
%We define the level of a node $P$, $\ell(P)$, as the number of deterministic elements in it.
Every edge is between two nodes at adjacent levels, the parent node being one level smaller than the child node.
\end{definition}

\begin{figure}[!tb]
\includegraphics[width = 0.43\textwidth]{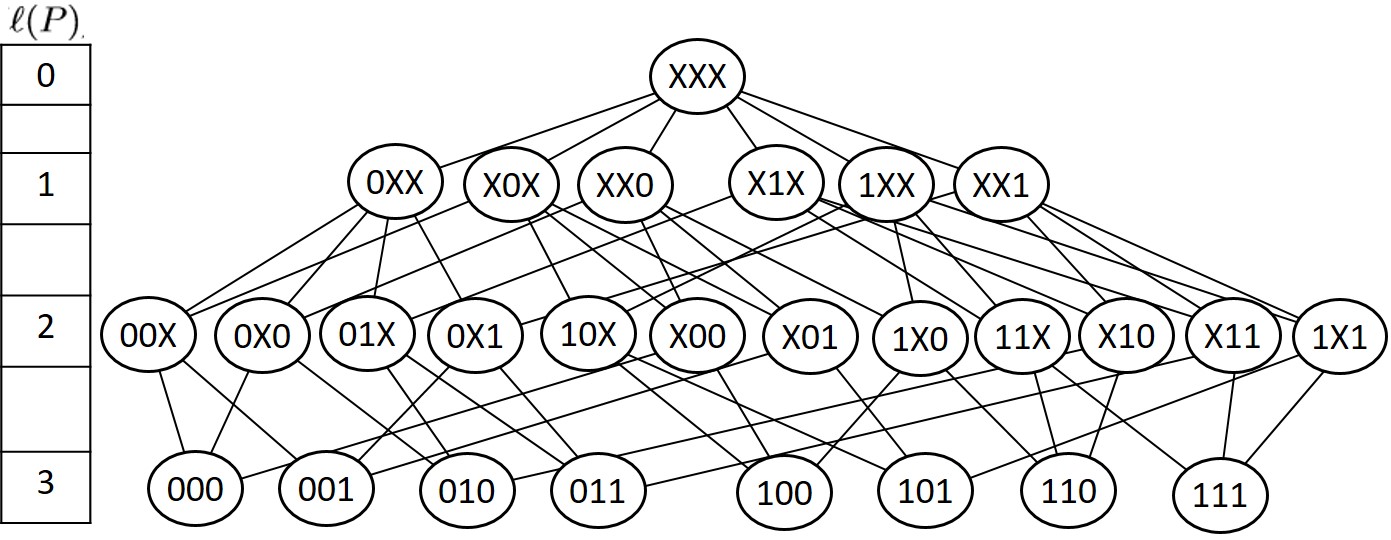}
        \vspace{-2.5mm}
        \caption{\footnotesize The pattern graph for three binary attributes}
        \label{fig:pg1}
        \vspace{-6mm}
\end{figure}

 Figure~\ref{fig:pg1} shows the pattern graph for Example~\ref{ex:1}.
The value of $P[i]$ for a node in the pattern graph is either $X$ or one of the values the corresponding attribute can take. Hence, the total number of nodes in a pattern graph defined over $d$ attributes \del{with cardinalities $c$,} \new{is $c^+_\mathcal{A} = \Pi_{k=1}^d(c_k+1)$}. For instance, the pattern graph in Figure~\ref{fig:pg1} contains $(2+1)^3 = 27$ nodes.
Any pattern graph has only one node $XX\cdots X$ at level 0.
%The most general pattern that every tuple in the dataset matches it is $XXX$.
In Figure~\ref{fig:pg1}, the patterns at level 1 are directly connected to $XXX$ as those are its children. The pattern $0XX$, the left-most node at level 1, is connected to the patterns $00X$, $0X0$, $01X$, and $0X1$ at level 2.
That is because, those have exactly one less $X$ and their value for the first attribute is $0$.
\new{
As explained in \S~\ref{sec:pre}, the number of value combinations matching $P$ (with non-deterministic attributes $\mathcal{A}_P$) is $c_{\mathcal{A}_P} = \Pi_{\forall A_j\in\mathcal{A}_P} c_j$.
The number of non-deterministic attributes of a pattern $P$ with level $\ell(P)$ is $d-\ell(P)$.
% The number of value combinations matching each pattern $P$ is $c^{d-\ell(P)}$.
For example, in Figure~\ref{fig:pg1},
every pattern at level 1 contains $3-1=2$ non-deterministic elements.
Since the cardinality of all attributes is $c_i=2$, the number of value combinations matching each pattern at level 1 is $2\times 2 = 4$. 
Hence, in general, the patterns with smaller levels are more general, i.e., more value combinations match them.
Each node at level $\ell$ contains $d-l$ non-deterministic elements.
There are ${d \choose \ell}$ such combinations in all.
The deterministic elements can take any value in the cardinality of the corresponding attribute.
Hence, for the special case where all attributes have the same cardinality $c_i=c$, total number of nodes at level $\ell$ are ${d \choose \ell}c^\ell$.
}
For example, in Figure~\ref{fig:pg1}, there are ${3 \choose 1} 2^1 = 6$ nodes at level 1 and ${3 \choose 2} 2^2 = 12$ nodes at level 2.
\new{
The node of a pattern $P$ in the pattern graph has $\sum_{\forall A_i\in \mathcal{A}_P}c_i$ edges to the nodes at level $\ell(P)+1$.
If all attributes have equal cardinalities of $c_i=c$, each node at level $\ell$ has $c(d-\ell)$ edges to nodes at level $\ell+1$.
Hence, the total number of edges in such a graph is:}
$$\sum\limits_{\ell=0}^{d-1} c(d-\ell) {d \choose \ell} c^\ell = c\times d\times (c+1)^{d-1}$$
This is confirmed in Figure~\ref{fig:pg1}, where there are totally 54 edges.

\submit{
\revision{
\subsection{\patternbreaker\ (the top-down algorithm)
and \patterncombiner\ (the bottom-up algorithm)}
In this section we propose the top-down and bottom-up algorithms \patternbreaker\ and \patterncombiner, that lead to the design of our final MUP identification algorithm, \deepdiver.
Due to the space limitations, we provide a sketch of these algorithms and refer the reader to the technical report~\cite{techreport} for further details.

\patternbreaker\ starts from the general patterns at the top of the pattern graph and moves down by breaking them down to more specific ones.
It uses the {\em ``monotonicity''} property of coverage to prune some parts of the pattern graph.
That is, if a pattern $P$ is uncovered, all of its children and descendants (the nodes at level greater than $\ell(P)$ that have a path to it) are also uncovered.
Also, none of those children and descendants can be a MUP, even if it has a parent that is covered.
Hence, this subgraph of the pattern graph can immediately get pruned.
In the BFS traversal of the pattern graph, the algorithm enforces Rule 1, states in the following, to ensure that each MUP candidate is generated exactly once.

\stitle{Rule 1.}
A node $P$ with the coverage more than the threshold $\tau$, generates the candidate nodes at level $\ell(P)+1$ by replacing the non-deterministic elements in the right-hand side of its right-most deterministic element with an attribute value.

\begin{theorem}\label{th:rule1completeness}
Enforcing Rule 1 guarantees that each MUP candidate is generated exactly once.
\end{theorem}

\begin{figure}[!t]
\centering
\includegraphics[width = 0.42\textwidth]{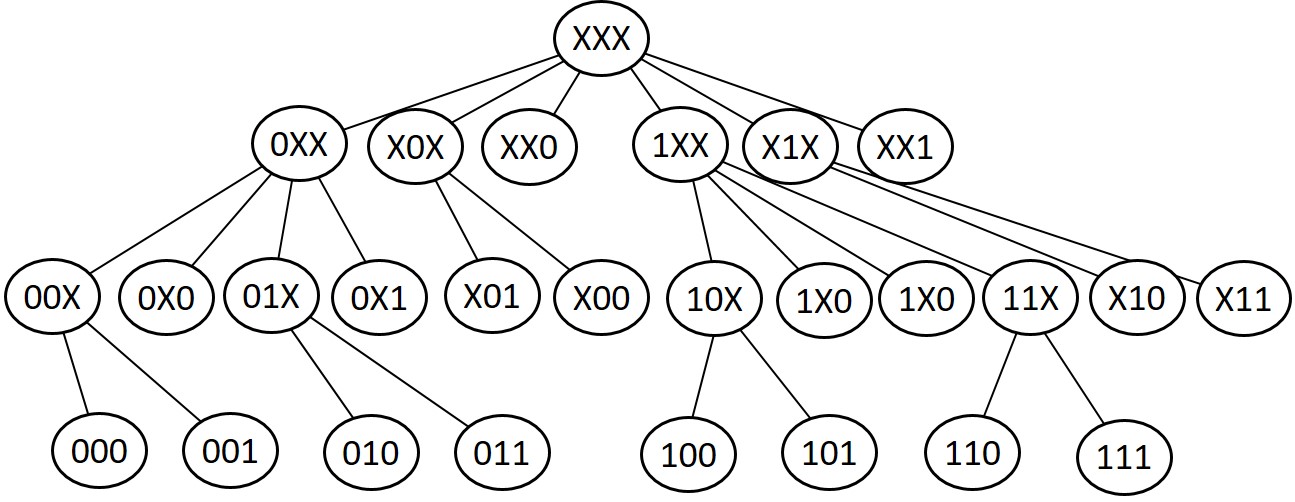}
		\vspace{-2.5mm}
        \caption{\footnotesize Tree transformation for Figure~\ref{fig:pg1} based on Rule 1.}
        \label{fig:pt1}
        \vspace{-6mm}
\end{figure}

By enforcing Rule 1, the pattern graph is transformed to a {\em tree}. 
For example, Figure~\ref{fig:pt1} shows the corresponding tree (generated by following Rule 1) for the pattern graph of Figure~\ref{fig:pg1}.

% A problem with \patternbreaker\ is that it traverses over the {\em covered} regions of the graph while we are looking for the uncovered nodes as the output. This makes its running time dependent on the size of the covered part of the graph.
% Moreover, the algorithm may mistakenly generate a large number of candidate nodes from the pruned regions.

The other algorithm \patterncombiner, performs a bottom-up traversal of the pattern graph. 
It uses an observation that the coverage of a node at level $\ell$ of the pattern graph can be computed using the coverage values of its children at level $\ell+1$.
The algorithm also uses the monotonicity of the coverage to prevent the complete traversal of the graph.
That is, the coverage of a node is not less than the coverage of any of its children. 
It starts from the most specific patterns, i.e., the patterns at level $d$ of the graph, computes the coverage of each by passing over the data once. The algorithm then, keeps combining the uncovered patterns at each level to get the coverage of the candidate nodes at level $\ell-1$. The uncovered nodes at level $\ell$ that all of their parents at level $\ell-1$ are covered are MUPs.
\patterncombiner\ transforms the pattern graph to a {\em forest}, in order to make sure every node is generated once. Rule 2 guarantees the transformation.
Figure~\ref{fig:pc1} shows the transformation of Figure~\ref{fig:pg1} to a forest, based on Rule 2.

\stitle{Rule 2.} 
A node $P$ with the coverage less than the threshold $\tau$, generates the candidate nodes at level $\ell(P)-1$ by replacing the deterministic elements with value 0 in the right-hand side of its right-most non-deterministic element with $X$.\footnote{\scriptsize \revision{Note that this rule is not specific to the binary attributes. All we require is that one of the values of each attribute is mapped to 0.}}

\begin{figure}[!t]
\centering
\includegraphics[width = 0.40\textwidth]{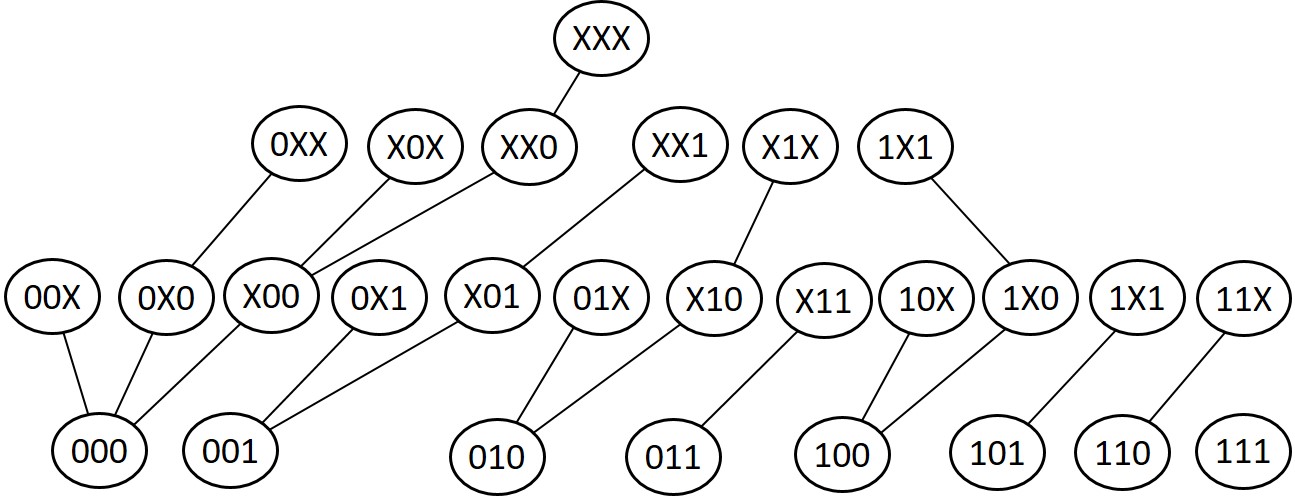}
        \vspace{-3mm}\caption{\footnotesize Forest transformation for Figure~\ref{fig:pg1} based on Rule 2.}
        \label{fig:pc1}
        \vspace{-7mm}
\end{figure}

\begin{theorem}\label{th:rule2completeness}
Enforcing Rule 2 guarantees each MUP candidate is generated exactly once.
\end{theorem}

}
}
\techrep{
\subsection{\patternbreaker: The top-down algorithm}
The root of the pattern graph, level 0, contains a single node, the most general pattern that matches every value combination.
The lower levels of the graph contain more specific patterns that match fewer combinations of values. 
%In this section, we propose 
The \patternbreaker\ algorithm starts from the general patterns at the top of the pattern graph and moves down by breaking them down to more specific ones. 
It uses the {\em ``monotonicity''} property of coverage to prune some parts of the pattern graph.
That is, if a pattern $P$ is uncovered, all of its children and descendants (the nodes at level greater than $\ell(P)$ that have a path to it) are also uncovered.
Also, none of those children and descendants can be a MUP, even if it has a parent that is covered.
Hence, this subgraph of the pattern graph can immediately get pruned.

The top-down BFS traversal of the pattern graph starts from the top of the graph and uses a queue to check the nodes level by level.
For every node in the queue, if its coverage is greater than the threshold, it adds each of its unvisited children to the queue.
This method, however, generates each node multiple times by all of its parents.
Therefore, \patternbreaker\ uses the following rule:
%to make sure every candidate pattern is generated exactly once.

\stitle{Rule 1.}
A node $P$ with the coverage more than the threshold $\tau$, generates the candidate nodes at level $\ell(P)+1$ by replacing the non-deterministic elements in the right-hand side of its right-most deterministic element with an attribute value.

\begin{theorem}\label{th:rule1completeness}
Enforcing Rule 1 guarantees that each MUP candidate is generated exactly once.
\end{theorem}
\submit{The proof of Theorem~\ref{th:rule1completeness} is in our technical report \cite{techreport}.}
\techrep{
\begin{proof}
Every node $P$ (other than the root) in the pattern graph
has exactly one parent node $P'$ that is in charge of generating it based on Rule 1. This parent can be found by replacing the right-most deterministic element of $P$ with $X$.

Also, based on Definition~\ref{def:mup}, since the coverage of all of the parents of every MUP $P\in \mathcal{M}$ is more than the threshold $\tau$, the coverage of the generator of $P$ based on Rule 1 (and every node in the path to root) is also more than the threshold. Therefore, following Rule 1 guarantees each MUP candidate once and only once.
\end{proof}
}

\begin{figure}[!t]
\centering
\includegraphics[width = 0.42\textwidth]{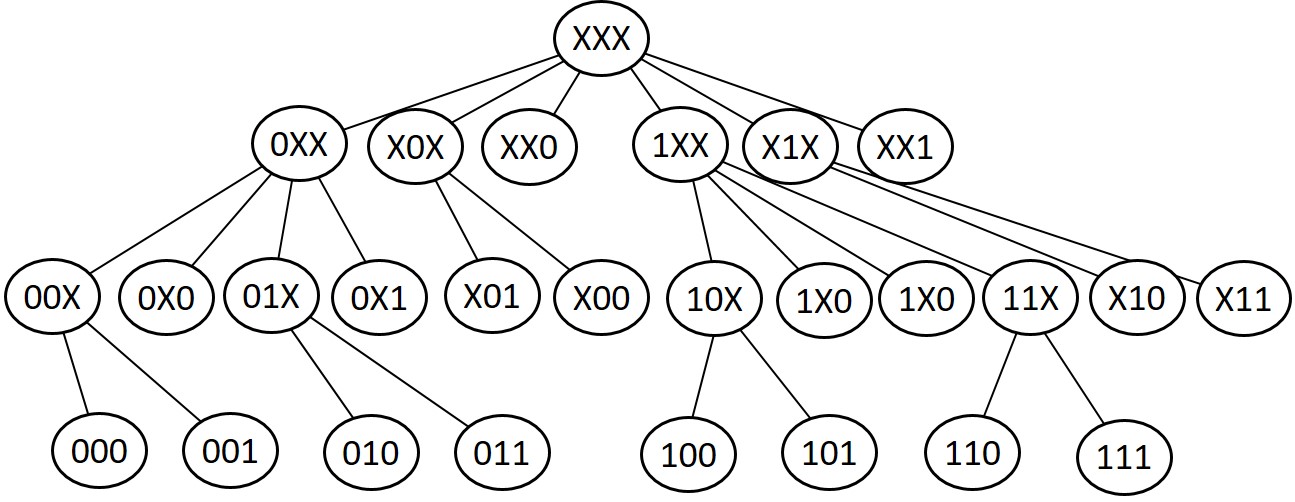}
		\vspace{-2.5mm}
        \caption{\footnotesize Tree transformation for Figure~\ref{fig:pg1} based on Rule 1.}
        \label{fig:pt1}
        \vspace{-6mm}
\end{figure}

By enforcing Rule 1, every node in the pattern graph gets generated at most once, and hence, the graph is transformed to a {\em tree}. 
For example, Figure~\ref{fig:pt1} shows the corresponding tree (generated by following Rule 1) for the pattern graph of Figure~\ref{fig:pg1}.
In this figure, for instance, consider the node $0XX$ while considering Example~\ref{ex:1}.
$cov(0XX)=5\geq \tau=1$. The right-most deterministic element is element 1 with value $0$.
Thus, it replaces the non-deterministic elements at positions 2 and 3 one at a time, and generates the nodes $0X0$, $0X1$, $00X$, and $01X$.
Similarly, the node $X1X$ generates the nodes $X10$ and $X11$. That is because the right-most non-deterministic element in this pattern is element 2 with value $1$. So it replaces the $X$ at position 3, with attribute values.

\begin{algorithm}[!ht]
\footnotesize
\caption{\footnotesize \patternbreaker \\
		 {\bf Input:} Dataset $\mathcal{D}$ with $d$ attributes having cardinalities $c$ and threshold $\tau$\\
		 {\bf Output:} Maximal uncovered patterns $\mathcal{M}$
		}
\begin{algorithmic}[1]
\label{alg:pb}
	\STATE $Q = \{ XX\cdots X\}$ {\scriptsize \tt // start from the root}
    \STATE $\mathcal{M} = \{\}$; $Q_p=\{\}$
    \FOR{{\scriptsize \tt /* each level of the graph */} $l=0$ to $d$}
    	\STATE {\bf if} $Q$ is empty {\bf then break}
        \STATE $Q_n=\{\}$
    	\FOR{$P\in Q$}
        	\STATE flag = false
        	\FOR{$P'$ in parents($P$)}
        		\STATE {\bf if} $P'\notin Q_p$ or $P'\in\mathcal{M}$ {\bf then} flag = true; {\bf break}
        	\ENDFOR
        	\STATE {\bf if} flag {\bf then continue}
        	\STATE $cnt = cov(P,\mathcal{D})$ \label{line:cov-ref}
            \IF{$cnt < \tau$}
        		\STATE add $P$ to $\mathcal{M}$
        	\ELSE
        	\STATE add children of $P$ based on Rule 1 to $Q_{n}$
        	\ENDIF
    	\ENDFOR
        \STATE $Q_p = Q$; $Q = Q_n$
    \ENDFOR
    \STATE {\bf return} $\mathcal{M}$
\end{algorithmic}
\end{algorithm}

Starting from the root, the algorithm moves level by level, checking the candidate patterns at each level. In addition to the list of nodes at the current level, it maintains the nodes at previous level, and constructs the nodes at the next level.
For every candidate node at the current level, it first checks if any of its parents is uncovered; if so, it marks the node as uncovered without computing its coverage.
Otherwise, the node is added to the set of MUPs if its coverage is less than the threshold. If the node satisfies the coverage threshold, its children, while enforcing Rule 1, are added to the list of candidate nodes of the next level.
We use inverted indices~\cite{invertedindex1} for computing the coverage of a pattern.
Further details about the efficient coverage computation are provided in Appendix~\ref{ap:inverted}.
Using the monotonicity property of the pattern graph, \patternbreaker\
\techrep{(Algorithm~\ref{alg:pb})}
tries to gain performance by pruning some parts of the graph.
%However, as we shall experiment in \S~\ref{sec:exp}, it has a several issues that makes it less efficient in practice.
A problem with \patternbreaker\ is that it traverses over the {\em covered} regions of the graph while we are looking for the uncovered nodes as the output. This makes its running time dependent on the size of the covered part of the graph.
%Hence, its performance is independent given the exponential size of the pattern graph, is exponential to the output.
Moreover, the algorithm may mistakenly generate a large number of candidate nodes from the pruned regions.
For example, in Figure~\ref{fig:pt1}, consider a case where $\tau=1$ and the dataset does not have any items matching the pattern XX1, but contains the items $t_1=000$ and $t_2=010$.
In this example, XX1 is a MUP. \revision{Following Rule 1}, \patternbreaker\ generates the pattern 0XX, and its children in Figure~\ref{fig:pt1} as both $t_1$ and $t_2$ match it (i.e., its coverage is more than the threshold).
One of its children (0X1)  has the coverage 0 (below the threshold), yet is not a MUP, because it is dominated by the MUP XX1.

\subsection{\patterncombiner: The bottom-up algorithm}
As observed in the analysis of \patternbreaker,
its problem is that it explores the covered regions of the pattern graph while the objective is to discover the uncovered ones. As a result, for the cases that there are only a few uncovered patterns, it needs to explore a large portion of the exponential-size graph.
Also, for each candidate node that it cannot prune, the algorithm needs to compute the coverage.

Therefore, we propose \patterncombiner, an algorithm that performs a bottom-up traversal of the pattern graph. 
It uses an observation that the coverage of a node at level $\ell$ of the pattern graph can be computed using the coverage values of its children at level $\ell+1$.
Consider a pattern $P$ and a non-deterministic element $P[i]$ in it. Note that the pattern can contain more than one non-deterministic element.
Consider the children of $P$ in which element $i$ is deterministic.
These children create $c_i$ {\em disjoint partitions} of the matches of $P$ ($c_i$ is the cardinality of attribute $A_i$). Thus, one can compute the coverage of $P$ as the sum of the coverages of those children.
{\em E.g.}, in Figure~\ref{fig:pg1}, $cov(1XX) = cov(1X0) +cov(1X1)$.

The algorithm \patterncombiner\ also uses the monotonicity of the coverage to prevent the complete traversal of the graph.
That is, the coverage of a node is not less than any of the coverage of any of its children. 
It starts from the most specific patterns, i.e., the patterns at level $d$ of the pattern graph, computes the coverage of each by passing over the data once. The algorithm then, keeps combining the uncovered patterns at each level to get the coverage of the candidate nodes at level $\ell-1$. The uncovered nodes at level $\ell$ that all of their parents at level $\ell-1$ are covered are MUPs.

\patterncombiner\ transforms the pattern graph to a {\em forest}, in order to make sure every node is generated once. Rule 2 guarantees the transformation.

\stitle{Rule 2.} 
A node $P$ with the coverage less than the threshold $\tau$, generates the candidate nodes at level $\ell(P)-1$ by replacing the deterministic elements with value 0 in the right-hand side of its right-most non-deterministic element with $X$.\footnote{\scriptsize \revision{Note that this rule is not specific to the binary attributes. All we require is that one of the values of each attribute is mapped to 0.}}

For example, consider the node $P=X01$ in Figure~\ref{fig:pg1}. Its element $P[1]$, with value $X$, is the right-most non-deterministic element in $P$. The only deterministic element with value 0 in the right hand-side of element 1 is element 2.
Therefore, this node generates the node $XX1$ in the bottom-up traversal of the graph.
As another example, consider the node with pattern $P=000$. It does not have any non-deterministic element and therefore, $0$ is considered as the index of the right-most $X$. elements 1, 2, and 3 have values 0 and therefore, this node generates the nodes $00X$, $0X0$, and $X00$.
Figure~\ref{fig:pc1} shows the transformation of Figure~\ref{fig:pg1} to a forest, based on Rule 2.

\begin{figure}[!t]
\centering
\includegraphics[width = 0.40\textwidth]{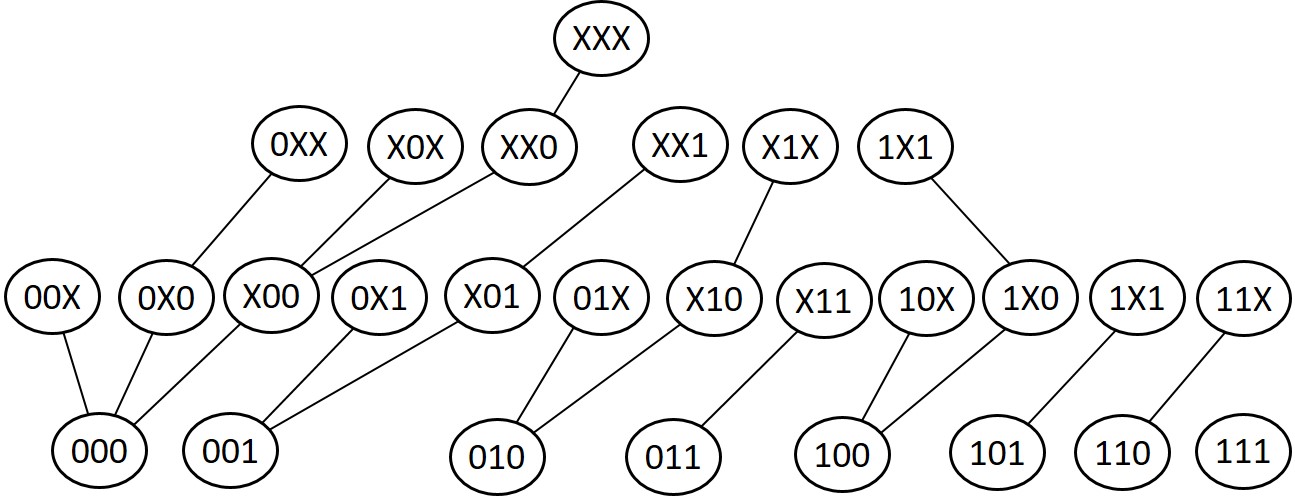}
        \vspace{-3mm}\caption{\footnotesize Forest transformation for Figure~\ref{fig:pg1} based on Rule 2.}
        \label{fig:pc1}
        \vspace{-7mm}
\end{figure}

\begin{theorem}\label{th:rule2completeness}
Enforcing Rule 2 guarantees each MUP candidate is generated exactly once.
\end{theorem}
\begin{proof}
The proof is based on the fact that every non-leaf node $P$ in the pattern graph has exactly one child node $P'$ that is in charge of generating it based on Rule 2.
This child can be found by replacing the right-most non-deterministic element of $P$ with $0$. 

Also, based on the monotonicity of coverage, all of the children of every MUP $P\in\mathcal{M}$ has the coverage less than $\tau$. Therefore, the coverage of the generator of $P$ (and every node in the path to the starting leaf node) based on Rule 2 is also less than the threshold. Consequently, following Rule 2 guarantees to generate each MUP candidate exactly once.
\end{proof}

\patterncombiner\ prunes a branch once it finds a covered node.
Due to the monotonicity of coverage if a child node is covered, all of its parents (and ancestors) are also covered. Therefore, a node is not pruned by the algorithm unless it is covered. This, together with Theorem~\ref{th:rule2completeness}, provide the assurance that all MUPs are discovered by \patterncombiner.

Starting from the bottom of the pattern graph, the algorithm first iterates over the dataset to compute the coverage of each node at level $\ell = d$.
Maintaining the record of the uncovered nodes at each level, the algorithm generates the nodes at level $\ell-1$ by applying Rule 2 on the uncovered nodes of level $\ell$.
For each node $P'$ at level $\ell-1$, \patterncombiner\ finds a set of nodes at level $\ell$ that create disjoint partitions of the matches of $P'$.
To do so, it finds the right-most non-deterministic element $i$ in $P'$ and considers the children of $P'$ that have a value on element $i$. The coverage of $P'$ is the summation of the coverages of those nodes.
If the coverage $P'$ is less than $\tau$ its record is maintained as a candidate MUP.
After finding the uncovered nodes at level $\ell-1$, the algorithm adds the MUPs at level $\ell$ to the output set. 
The pseudocode of \patterncombiner\ is in \submit{the technical report \cite{techreport}}\techrep{Algorithm~\ref{alg:pc}}.

\techrep{
\vspace{3mm}
\begin{algorithm}[!t]
\footnotesize
\caption{\footnotesize \patterncombiner \\
		 {\bf Input:} Dataset $\mathcal{D}$ with $d$ attributes having cardinalities $c$ and threshold $\tau$\\
		 {\bf Output:} Maximal uncovered patterns $\mathcal{M}$
		}
\begin{algorithmic}[1]
\label{alg:pc}
	\STATE count = {\it new} hash()
	\FOR{$\forall P\in \{ v_1 v_2\cdots  v_d ~|~ v_i\in c[A_i]\}$} 
    	\STATE cnt$= cov(P,\mathcal{D})$
        \STATE {\bf if} cnt$<\tau$ {\bf then} count$[P] = $ cnt
    \ENDFOR
	\STATE {\bf if} count is empty {\bf then return} $\emptyset$
    \FOR{$\ell = 0$ to $d$}
    	\STATE nextCount = {\it new} hash()
        \FOR{$P$ in count.keys}
        	\STATE $\mathcal{P}' =$ generates nodes at $\ell-1$ based on Rule 2 on $P$
            \FOR{$P'\in\mathcal{P}'$}
            	\STATE $i = $ the index of right-most $X$ in $P'$
                \STATE $\mathcal{P}'' = \{ P''|\forall j\neq i\,:P''[j]=P' \mbox{ and } P''[i]\in c[A_i] \}$
            	\STATE $cnt = \sum\limits_{\forall P''\in \mathcal{P}''} (\mbox{count}[P'']$ {\bf if} $P''\in$ count.keys {\bf else} $\tau)$
        		\STATE {\bf if} $cnt <\tau$ {\bf then} nextCount$[P'] = $ cnt
            \ENDFOR
        \ENDFOR
        \FOR{$P$ in count.keys}
            \IF{parents($P)\cap$ nextCount.keys $=\emptyset$}
            	\STATE add $P$ to $\mathcal{M}$
            \ENDIF
        \ENDFOR
        \STATE {\bf if} nextCount is empty {\bf then break} 
        \STATE count = nextCount
    \ENDFOR
    \STATE {\bf return} $\mathcal{M}$
\end{algorithmic}
\end{algorithm}
}
}

\subsection{\deepdiver: Fast search space pruner}\label{subsec:deepdiv}
\patternbreaker\ traverses over the covered regions of the pattern graph before it visits the uncovered patterns.
% \jin{before it visits the uncovered patterns}. 
Therefore, it does not perform well when a large portion of the pattern graph is covered.
Conversely, \patterncombiner\ traverses over the uncovered nodes first; so it will not perform well if most of the nodes in the graph are uncovered.
When most MUPs are in the middle of the graph, both algorithms do poorly because they have to traverse about half of the graph.
%may need to traverse a large portion of the pattern graph before they find the MUPs.
%As a result, depending on where the MUPs are located, both \patternbreaker\ and \patterncombiner\ may lose their efficiencies.
In this subsection, we propose \deepdiver, an algorithm that 
%repeatedly switches direction to try to do better than both \patternbreaker\ and \patterncombiner.  The idea is 
tries to quickly identify some MUPs and use them to prune the search space.

\begin{comment}
%\jin{did we claim our problem is NP-complete?} \abol{it is not. it is \#P-complete; but we did not discuss it} 

As discussed previously, the search for MUPs is a challenging problem. Any user running either of the above two algorithms might need to wait for infinite time before the completion of the search when the dataset is extremely high-dimensional (there are many attributes of interested). A natural way to alleviate this problem is to allow the algorithm to terminate early and return partial results. In this case, it is reasonable to expect a good algorithm to have \textit{Anytime Property}---discovering as many MUPs as possible within the limited time. 

An extra bonus of Anytime Property is that finding MUPs early helps boosting the search speed. 
The monotonicity property creates an opportunity of pruning the search space: all ancestor or descendant patterns of a given MUP can never be MUP and hence should be pruned.  

However, neither \patternbreaker\ or \patterncombiner\ exhibits Anytime Property. \patternbreaker\ tends to traverse over the covered regions of the pattern graph before it visits the uncovered patterns. Although, \patterncombiner\ starts off in the uncovered regions initially, the nodes being visited early are at the bottom of the pattern graph and often not close to MUPs. The algorithm can be  trapped in a ``swamp'' of uninteresting uncovered patterns for quite a while before it visits the uncovered patterns that are actually MUPs.
\end{comment}

The monotonicity property creates an opportunity to prune the search space: {\em none of the ancestors or descendants of a given MUP can be MUP}.
\patternbreaker\ tends to traverse level by level over the covered regions of the pattern graph before it visits the uncovered patterns.
As a result, the moment it reaches out to a MUP, it already has visited its ancestors and does not take the advantage of pruning the nodes dominating the MUPs.
\patterncombiner, on the other hand, starts off in the uncovered regions; initially, the nodes being visited early are at the bottom of the pattern graph. It gradually moves up level by level until it hits the MUPs. Therefore, when the MUPs are discovered the descendants have already been visited and, as a result, \patterncombiner\ does not take the advantage of pruning the nodes dominated by MUPs.

With the above observations, we propose \deepdiver, 
%(Algorithm~\ref{algo:deepdiver} in Appendix~\ref{ap:pc}), 
a search algorithm that tends to quickly find MUPs, and use them to limit the search space by pruning the nodes both dominating and dominated by the discovered MUPs.
Since each MUP is the child of a covered node, instead of scanning through the covered/uncovered patterns level by level, \deepdiver\ 
takes a path down to find an uncovered node.

\begin{algorithm}[t]
\footnotesize
\caption{\revision{
        \footnotesize \deepdiver \\
		 {\bf Input:} Dataset $\mathcal{D}$ with $d$ attributes having cardinalities $c$ and threshold $\tau$\\
		 {\bf Output:} Maximal uncovered patterns $\mathcal{M}$
		 }
		}\label{algo:deepdiver}
\begin{algorithmic}[1]
	\STATE Let $S$ = an empty stack 
    \STATE push $X\cdots X$ to $S$
%     \STATE {\bf return} $\mathcal{M}$
    \WHILE{$S$ is not empty} 
    \STATE{$P$ = pop a node from $S$}  \label{algo:deepdiver:divedownstart}
    \STATE $uncoveredFlag =$ a flag indicating if $P$ is uncovered
    \IF{ $P$ is dominated by $\mathcal{M}$} \label{algo:deepdiver:prunestart}
    \STATE continue \label{algo:deepdiver:pruneend}
    \ELSIF {$P$ dominates $\mathcal{M}$} \label{algo:deepdiver:prune2start}
    \STATE $uncoveredFlag = true$ \label{algo:deepdiver:prune2end}
    \ELSE 
    \STATE{$cnt = cov(P,\mathcal{D})$}   
    \STATE $uncoveredFlag = cnt < \tau$ 
    \ENDIF \label{algo:deepdiver:divedownend}
    \IF{$uncoveredFlag$ is $true$} \label{algo:deepdiver:goupstart}
    \STATE Let $S'$ = an empty stack
    \WHILE{$S'$ is not empty} 
    \STATE $P'$ = pop a node from $S'$
    \STATE $\mathcal{P}' =$ generates parent nodes of $P'$ by replacing one deterministic cell with $X$.
            \FOR{$P''\in\mathcal{P}'$}
            \STATE $cnt' = cov(P'',\mathcal{D})$
            \STATE {\bf if} $cnt' < \tau$ {\bf then} push $P''$ to $S'$; {\bf break}
% 			\IF {$cnt' < \tau$}
%             \STATE push $P''$ to $S'$
%             \STATE break
%             \ENDIF
             \ENDFOR
            \STATE add $P$ to $\mathcal{M}$ \label{algo:deepdiver:findmup}
    
    \ENDWHILE\label{algo:deepdiver:goupend}
    \ELSE \label{algo:deepdiver:keepdivingstart}
    
    \STATE $Q =$ generate nodes on $P$ and $c$ based on Rule 1
    \STATE push $Q$ to $S$
    \ENDIF \label{algo:deepdiver:keepdivingend}
%     \IF{$p$ dominates $\mathcal{M}$}
%     \STATE {}
%     \ELSIF{$p$ is dominated by $\mathcal{M}$}
%     \STATE{}
%     \ENDIF
    \ENDWHILE
    \STATE {\bf return} $\mathcal{M}$
\end{algorithmic}
\end{algorithm}

Initially, \deepdiver (Algorithm~\ref{algo:deepdiver}), following a DFS strategy, takes a path down until it reaches into an uncovered region in the graph.
% (line~\ref{algo:deepdiver:divedownstart}-line~\ref{algo:deepdiver:divedownend} and line~\ref{algo:deepdiver:keepdivingstart}-line~\ref{algo:deepdiver:keepdivingend}).
However, the discovered uncovered pattern is not necessarily a MUP, as some of its other parents (other than its generator) might also be uncovered.
For instance in Example~\ref{ex:1}, assume that in the first iteration, the algorithm take the path XXX $\rightarrow$ X0X $\rightarrow$ 10X.
The nodes XXX and X0X are covered, but 10X is not.
Still the uncovered node 10X is not a MUP as it has the uncovered parent 1XX. 
Therefore, after finding an uncovered node, \deepdiver\
changes direction and starts moving up to find a MUP.
To do so, it checks the parents of the current node to see if any of them are uncovered. If there exists such a parent, it moves to the parent and continues until it finds a MUP.
% (line~\ref{algo:deepdiver:goupstart}-line~\ref{algo:deepdiver:goupend}).
Upon discovering a MUP, \deepdiver\ prunes all of its ancestors and descendants, and continues the search for other MUPs in the regions that are still not pruned.
The algorithm stops when all of the nodes in the pattern graph are pruned.

Here we extend the notion pattern dominance to MUP dominance, as follows:

\begin{definition}[MUP Dominance]\label{defn:dominance}
Given a pattern $P$ and a set of MUPs $\mathcal{M}$, $P$ is dominated by $\mathcal{M}$, if there exists a pattern $P'\in\mathcal{M}$ such that $P$ is dominated by $P'$.
Similarly, $P$ dominates $\mathcal{M}$, if there exists a pattern $P'\in\mathcal{M}$ such that $P$ dominates $P'$.
\end{definition}

Based on the Definition~\ref{defn:dominance}, a node being dominated by MUPs is not a MUP. 
\deepdiver\ uses this property to limit the search space by pruning all descendants of the MUPs. Similarly, the nodes that dominate MUPs are out of the search space.
We use inverted indices for efficiently checking MUP dominance.
See \submit{the technical report \cite{techreport}}\techrep{Appendix~\ref{ap:dominancechecking}} for details.

\section{Coverage Enhancement}\label{sec:4}
\techrep{
So far in this paper, we discussed how to address Problem 1, by discovering the maximal uncovered patterns.
A natural question, after discovering the MUPs is how to collect the data based on them.
In this section, we seek to answer this question, i.e., Problem 2.
}

% \begin{figure}[!t]
% \centering
% \includegraphics[width = 0.3\textwidth]{figures/covcube1.jpg}
%         \vspace{-3mm}\caption{\footnotesize The matches of three patterns XX1: red, 0XX: green, and 20X: blue in a data set with ternary attributes $A_1$, $A_2$, and $A_3$.}
%         \label{fig:hs1}
%         \vspace{-4mm}
% \end{figure}

\definecolor{bblue}{HTML}{4F81BD}
\definecolor{rred}{HTML}{C0504D}
\definecolor{ggreen}{HTML}{9BBB59}
\definecolor{ppurple}{HTML}{9F4C7C}

\begin{figure*}[!ht]
    \begin{minipage}[t]{0.25\linewidth}
    \centering
    \vspace{-37mm}
    \includegraphics[width = 1.0\textwidth]{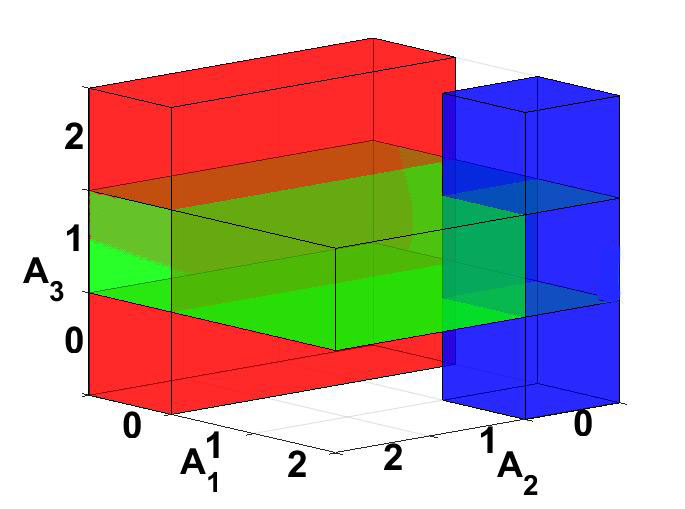}
        \vspace{-8mm}\caption{\footnotesize The matches of three patterns \revision{XX1: green, 0XX: red}, and 20X: blue in a data set with ternary attributes $A_1$, $A_2$, and $A_3$.}
        \label{fig:hs1}
    \end{minipage}
    \hfill
    \begin{minipage}[t]{0.3\linewidth}
        \centering
        \begin{tikzpicture}
        \begin{axis}[
        width=\textwidth,
        height=4cm,
            ymin=1,
            ymax=1250,
            xmin=-1,
            xmax=12,    
            xtick=data,
            ytick={250, 500, 750, 1000},
            ybar,
            xlabel={Level},
            y label style={at={(axis description cs:0.1,.5)}},
            ylabel={\# of MUPs},
            nodes near coords,
            every node near coord/.append style={font=\scriptsize},
            major x tick style = transparent,
            xticklabel style={font=\scriptsize},
            yticklabel style={font=\scriptsize},
            ymajorgrids=true
            ]
        \addplot [color=blue,fill=bblue] table[x=level, y=count, col sep=comma] {data/mupDist.csv};
        \end{axis}
        \end{tikzpicture}
         \vspace{-4mm}\caption{\footnotesize Distribution of MUPs in AirBnB dataset for $n=1000$ items and $d=13$ attributes, while $\tau=50$.}
        \label{fig:mupdist1}
    \end{minipage}
    \hspace{1mm}
    \begin{minipage}[t]{0.25\linewidth}
        \centering
        \vspace{-35mm}
        \includegraphics[width = 0.85\textwidth]{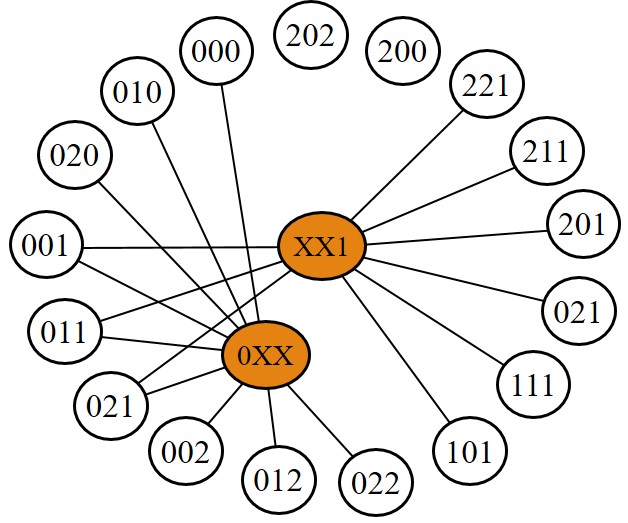}
        % \vspace{1mm}
        \caption{\footnotesize The bipartite graph for $\lambda=1$ for Figure~\ref{fig:hs1}.}
        \label{fig:bipartite}
    \end{minipage}
    \hspace{1mm}
    \begin{minipage}[t]{0.14\linewidth}
        \centering
        \vspace{-33mm}
        \small
        \begin{tabular}{|@{}c@{}|@{}c@{}|}
          \hline
          $P_1$ & XX01X \\ \hline
          $P_2$ & 1X20X \\ \hline
          $P_3$ & XXXX1 \\ \hline
          $P_4$ & 02XXX \\ \hline
          $P_5$ & XX11X \\ \hline
          $P_6$ & 111XX \\ \hline
          $P_7$ & X020X \\ \hline
    	\end{tabular}
        \vspace{6mm}
        \caption{\footnotesize MUPs of Example~\ref{ex:datacollection}.}
        \label{fig:ex2}
    \end{minipage}
    \hspace{-2mm}
    \vspace{-6mm}
\end{figure*}

\begin{comment}
\begin{figure}[!t]
\centering
\begin{tikzpicture}
\begin{axis}[
width=0.9\linewidth,
height=5cm,
    ymin=1,
    ymax=1250,
    xmin=-1,
    xmax=14,    
    xtick=data,
    ytick={250, 500, 750, 1000},
    ybar,
    xlabel={Level},
    ylabel={\# of MUPs},
    nodes near coords,
%     every node near coord/.append style={font=\tiny},
    major x tick style = transparent,
    ymajorgrids=true
%     bar width=2pt
    ]
\addplot [color=bblue,fill=bblue] table[x=level, y=count, col sep=comma] {data/mupDist.csv};
\end{axis}
\end{tikzpicture}
 \caption{The distribution of MUPs in an experiment with $n=1000$ items and $d=13$ attributes, while $\tau=50$.}
        \label{fig:mupdist1}
\end{figure}
\end{comment}

Every MUP represents a part of the value combinations space for which there are not enough observations in the dataset. 
For example, consider a dataset defined over three ternary attributes $A_1$, $A_2$, and $A_3$, in which MUPs are XX1, 0XX, and 20X.
Figure~\ref{fig:hs1} shows the matches for the patterns XX1, 0XX, and 20X, as the red, green, and blue cubes, respectively.
The more general MUPs show larger uncovered regions in the data. For example in Figure~\ref{fig:hs1}, the cube of the more general patterns XX1 and 0XX contain 9 combinations, whereas the one for the pattern 20X contains 3 combinations.
%The number of value combinations matching a pattern $P$ can simply be computed as $c^{d-\ell(P)}$, where $c$ is the cardinality of the attributes.
%For instance, in the example of Figure~\ref{fig:hs1}, the number of combinations matching the pattern XX1 match is $3^2 = 9$, while this number is $3^1 = 3$ for 20X.
%Note that span of patterns decrease exponentially by their levels.
While we may be willing to leave some small regions uncovered, we would like to cover at least the large ones.
For example, not having enough representatives in a dataset for single black males over the age of sixty may be less of a problem than not having enough black males.
%On the other hand, while the total number of MUPs may not be small, in practice there are a few MUPs that have a small level and appear close to the root of the pattern graph.
%For instance, 
Figure~\ref{fig:mupdist1} shows the distribution of the levels of the MUPs for a real experiment on our AirBnB dataset (c.f. \S~\ref{sec:exp}) with $n=1000$ items and $d=13$ attributes, while $\tau=50$.
There are several thousand MUPs in this setting. This indicates the high expense of covering them all.
However, as the distribution has a bell-curve shape, while most MUPs appear at levels 5 and 6, there is only one MUP at level one and less than forty in level two.
%This shows two attribute-values for which there are not enough observations in the data.
%Given that most of the uncovered patterns are relatively less harmful, we at least want to cover the general ones with small levels, as we do not want large spaces with not enough observations in the dataset.

% \begin{figure}[!t]
% \centering
% \includegraphics[width = 0.35\textwidth]{figures/mupsHistogram.png}
%         \caption{The distribution of MUPs in an experiment with $n=1000$ items and $d=13$ attributes, while $\tau=50$. \abol{we will reproduce this}}
%         \label{fig:mupdist1}
% \end{figure}

Data acquisition is usually costly.  If the data are obtained from some third party, there may be direct monetary payment.  If the data are directly collected, there may be a data collection cost.  In all cases, there is a cost to cleaning, storing, and indexing the data.
To minimize these costs, we would like to acquire as few additional tuples as possible to meet our coverage objective.

\revision{
Before discussing further technical details, we would like to emphasize the necessity of {\em human-in-the-loop} after the MUP discovery.
Not all the MUPs that are discovered are meaningful and some of them may even be invalid.
Therefore after the MUP discovery, a domain expert should evaluate and mark out the MUPs that are not problematic.
In addition, we require the expert to set up a {\em validation oracle} as a set of rules that identifies if a value combination is semantically correct or not.
For example, any value combination that contains \{gender=Male, isPregnant=True\} is semantically incorrect.

\begin{definition}[Validation Rule]\label{def:vr}
A validation rule is a set of pairs $\{\langle A_i, V_i\rangle,\cdots\}$, where $A_i$ is an attribute and $V_i$ is a set of values for $A_i$.
Given a pattern $P$ and a validation rule $R$, we say $P$ satisfies $R$, if $\forall \langle A_i, V_i\rangle\in R$: $P[i]\in V_i$.
\end{definition}

\begin{definition}[Validation Oracle]\label{def:vo}
A validation oracle contains a collection of validation of rules. Given a pattern $P$, the oracle returns true if $P$ satisfies {\em none} of its validation rules. It returns false otherwise.
\end{definition}
The human expert sets up the validation oracle by identifying the collection of validation rules.
Later on, in this section, we call the validation oracle to enforce the rules that result in the semantic appropriateness (validity) of the output of the coverage enhancement algorithm.
}

As formally defined in \S~\ref{sec:pre} as Problem 2, for a given value $\lambda$, our objective is to collect the minimum number of additional tuples such that after the data collection the maximum covered level of $\mathcal{D}$ is at least $\lambda$.
\new{
It is not enough to cover the MUPs with levels $\ell\leq\lambda$, we must cover all uncovered patterns (not necessarily maximal) with level $\ell(P)=\lambda$. 
We use $M_{\lambda}$ to refer to the set of uncovered patterns at level $\lambda$.
Finding this set is not difficult: details in \submit{\cite{techreport}}\techrep{Appendix~\ref{ap:ptc}}.
% \subsection{Simplified version of the problem}\label{subsec:simplified}
% As a simplified version of the problem, we first focus on covering the most general MUPs that have the smallest level in the pattern graph.
% Later on, in \S~\ref{subsec:general}, we show how this can be extended to any maximum coverage level of $\lambda$.
% Note that this is a special version of the problem in which $\lambda$ is equal to the level of the most general MUP in the current dataset.
}

We take Example~\ref{ex:datacollection} as a running example in this section.

\begin{example}\label{ex:datacollection}
Consider
a dataset $\mathcal{D}$ with 5 attributes $A_1$ to $A_5$ where $A_2$ and $A_3$ are ternary attributes while the other attributes are binary. Suppose the maximal uncovered patterns are as shown in Figure~\ref{fig:ex2}.
%     \begin{center}
%     \begin{tabular}{|c|c|}
%         \hline
%         $P_1$ & XX01X1X \\ \hline
%         $P_2$ & 1X20XXX \\ \hline
%         $P_3$ & XXXX111 \\ \hline
%         $P_4$ & 02XXX1X \\ \hline
%         $P_5$ & XX11XX1 \\ \hline
%         $P_6$ & 111XXXX \\ \hline
%         $P_7$ & X020XX1 \\ \hline
%     \end{tabular}
%     \end{center}
\end{example}

\noindent Let $\lambda$ be 2. Uncovered patterns of Example~\ref{ex:datacollection} with level $2$, i.e. $M_\lambda$, are $P_1$ to $P_6$.
Our objective is to cover all these patterns. 

If we use an alternative problem formulation, the set of patterns to cover may be different.  For example, if we wish to cover all patterns with value count of at least $v$, we must enumerate uncovered patterns that meet this value count criterion.  Once this (straightforward) enumeration is completed, thereafter the alternative problem formulation can be solved in exactly the same way.
A potential naive idea may be to acquire enough additional tuples separately for each pattern we are required to cover.
However, this ``solution'' acquires much more than the minimum data required, because each tuple may contribute to the coverage of multiple patterns.  What we need is to choose tuples carefully to find the minimum number needed to cover all the uncovered patterns of interest.
This problem translates to a hitting set~\cite{vazirani2013approximation} instance.

\vspace{-1mm}
\subsection{Transformation to hitting set}
\stitle{Hitting Set Problem:} Given a set $\mathcal{U}$ of elements and a collection $\mathcal{S}$ of non-empty subsets of $\mathcal{U}$, the objective is to find the smallest subset of elements $C\subseteq \mathcal{U}$ such that $\forall S\in\mathcal{S},~\exists~e\in C$ where $e\in S$.

The transformation is as follows:
\begin{itemize}[leftmargin=*]
\itemsep0em 
\item $\mathcal{U}$: The set of possible value combinations translates to the universe of items $\mathcal{U}$.
\item $\mathcal{S}$: Each uncovered pattern in $M_\lambda$ is the representative of the set of value combinations matching it. Hence, $\mathcal{S}$ is the collection of sets represented by the uncovered patterns.
\end{itemize}
%consider each pattern as the representative of the set of value combinations matching it.
%The objective is to find the minimum number of nodes that hit all the sets represented by the MUPs.
This can be viewed as a bipartite graph with the value combinations in the first part and the uncovered patterns in the second part.
There is an edge between a combination and a pattern if the combination matches the pattern.
%Each pattern $P$ is a associated with the weight $c^{X(P)}$.
The objective is to select the minimum number of nodes in the first part that hit all the patterns in the second part.
Figure~\ref{fig:bipartite} shows the bipartite graph for $\lambda=1$ in Figure~\ref{fig:hs1}.
%Due to the large number of value combinations, we circled the left partite around the patterns (right partite) and only drew the nodes.

\begin{comment}
\begin{figure}[!t]
\centering
\includegraphics[width = 0.33\textwidth]{figures/bipartite2.jpg}
        \caption{The bipartite graph for $\lambda=1$ for Figure~\ref{fig:hs1}.}
        \label{fig:bipartite}
\end{figure}
\end{comment}

While the hitting set problem is NP-complete, the greedy approach guarantees a logarithmic approximation ratio for it~\cite{vazirani2013approximation}.
%The greedy approximation algorithm for the hitting set, 
At every iteration, the greedy approximation algorithm selects the item (value combination)
that hits the maximum number of un-hit sets (patterns). It continues until all the sets get hit.
In Figure~\ref{fig:bipartite}, for instance, a run of the greedy algorithm picks 001 as it hits both patterns and then stops.

%The direct implementation of the algorithm, however, is not efficient here.  That is because, 
At every iteration, the algorithm needs to find the value combination that hits the maximum number of un-hit patterns. This is inefficient due to the exponentially large number of the value combinations and potentially exponential number of the patterns to hit.
%computes the sum of the weights for all the value combination, before it selects the maximum. As a result, the algorithm is expensive, in the order of $O(b c^d |\mathcal{MUP}|)$.
Hence, in the following, we develop an efficient implementation of the greedy algorithm.

% \begin{figure}[!t]
% \centering
% \begin{scriptsize}
%     \begin{tabular}{|@{}c@{}|@{}c@{}|@{}c@{}|@{}c@{}|@{}c@{}|@{}c@{}|@{}c@{}|}
%         \hline
%          & $P_1$& $P_2$& $P_3$& $P_4$& $P_5$& $P_6$ \\ \hline
%          $A_1 = 0$& 1& 0& 1& 1& 1&0 \\ \hline
%          $A_1 = 1$& 1& 1& 1& 0& 1&1 \\ \hline
%          $A_2 = 0$& 1& 1& 1& 0& 1&0 \\ \hline
%          $A_2 = 1$& 1& 1& 1& 0& 1&1 \\ \hline
%          $A_2 = 2$& 1& 1& 1& 1& 1&0 \\ \hline
%     \end{tabular}
% \end{scriptsize}
%         \vspace{-1mm}\caption{\footnotesize The inverted indices for values of attributes $A_1$ and $A_2$ in Example~\ref{ex:datacollection}.}
%         \label{fig:inv1}
%         \vspace{-5mm}
% \end{figure}

\begin{figure}[!t]
\begin{minipage}[t]{0.45\linewidth}
\centering
\vspace{-30mm}
\begin{scriptsize}
    \begin{tabular}{|@{}c@{}|@{}c@{}|@{}c@{}|@{}c@{}|@{}c@{}|@{}c@{}|@{}c@{}|}
        \hline
         & $P_1$& $P_2$& $P_3$& $P_4$& $P_5$& $P_6$ \\ \hline
         $A_1 = 0$& 1& 0& 1& 1& 1&0 \\ \hline
         $A_1 = 1$& 1& 1& 1& 0& 1&1 \\ \hline
         $A_2 = 0$& 1& 1& 1& 0& 1&0 \\ \hline
         $A_2 = 1$& 1& 1& 1& 0& 1&1 \\ \hline
         $A_2 = 2$& 1& 1& 1& 1& 1&0 \\ \hline
    \end{tabular}
\end{scriptsize}
        \vspace{10mm}\caption{\footnotesize The inverted indices for values of attributes $A_1$ and $A_2$ in Example~\ref{ex:datacollection}.}
        \label{fig:inv1}
        \vspace{-5mm}
\end{minipage}
\hfill
\begin{minipage}[t]{0.50\linewidth}
\centering
\includegraphics[width = \textwidth]{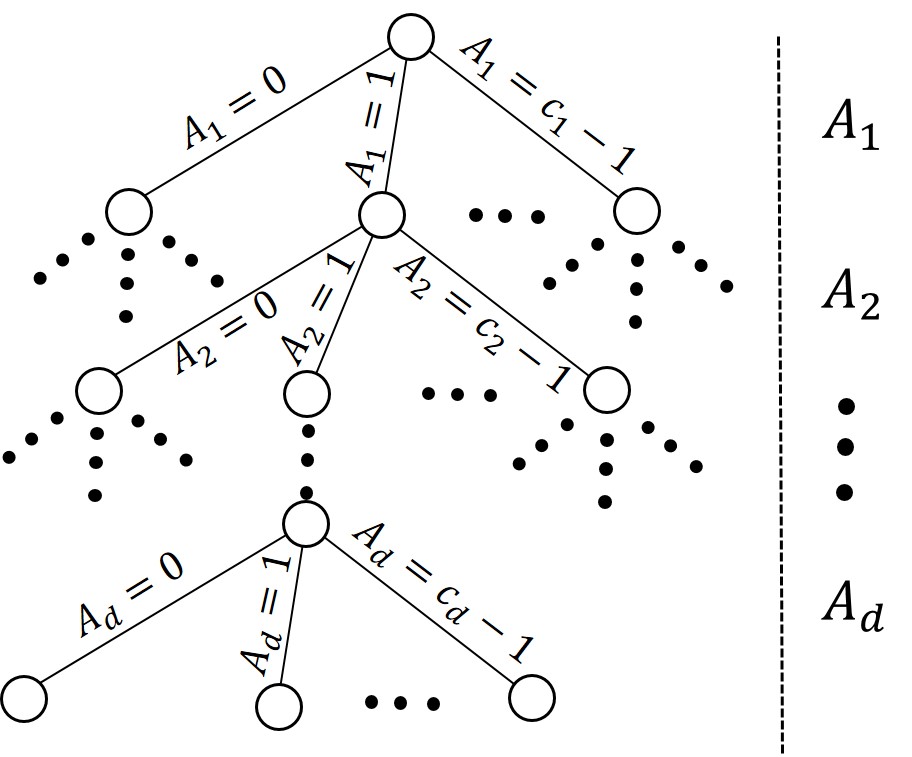}
        \vspace{-6mm}\caption{\footnotesize The tree data structure used for computing the number of patterns a value combination hits.}
        \label{fig:dctree}
\end{minipage}
\vspace{-6mm}
\end{figure}

\subsection{Efficient implementation of the greedy algorithm}\label{subsec:HS-efficient}
Consider the set of patterns we want to hit by the value combinations.
We use inverted indices to keep track of the uncovered patterns.
The $i$-th element of each pattern is either a value of $A_i$ or is non-deterministic.
For each attribute value $v_j$ of $A_i$, we create an inverted index to point to the patterns with either a non-deterministic element or an element with value $v_j$ in the $i$-th position.
We use this to filter out the patterns that do not match a value combination with value $v_j$ on $A_i$.
For example, Figure~\ref{fig:inv1} shows the indices for the values of $A_1$ and $A_2$ for $P_1$ to $P_6$ in Example~\ref{ex:datacollection}.
The first row shows the inverted index for $A_1=0$. All columns of the row, except $P_2$ and $P_6$ are 1; that is because a value combination having $A_1=0$ will not match $P_2$ or $P_6$, but still can match other patterns.
Having the inverted indices for the attribute values, we use a tree data structure and design a greedy threshold-based algorithm to find the value combination that hits the maximum number of remaining patterns.

Consider a full tree data structure (Figure~\ref{fig:dctree})
with depth $d$ (the number of attributes).
Let $m$ be the number of patterns that we want to hit.
The root has $c_1$ children, each showing a value for attribute $A_1$.
Similarly, every node at level $i$ has $c_i$ children, each representing an attribute value of $A_i$.
For instance, in Example~\ref{ex:datacollection}, the depth of the tree is $d=5$; the root node has two children standing for $A_1=0$ and $A_1=1$. Since $c_2=3$, each of these two nodes will have three children $A_2=0$, $A_2=1$, and $A_2=2$.
\submit{
Each path in the tree from the root to a leaf shows a value combination. In~\cite{techreport}, we show how to use the inverted indices for finding the number of remaining patterns a value combination hits.
}

% \begin{figure}[!t]
% \centering
% \includegraphics[width = 0.3\textwidth]{figures/DCtree.jpg}
%         \vspace{-2mm}\caption{\footnotesize The tree data structure used for computing the number of patterns a value combination hits.}
%         \label{fig:dctree}
%         \vspace{-3mm}
% \end{figure}
\techrep{
Consider a path in the tree from the root to a leaf; note that it shows a value combination.
We associate a bit vector of size $m$ with the root. Initially, all the values in the bit vector are 1, meaning that all patterns still remain to get hit.
To find the number of remaining patterns a value combination hits, one can follow the path from the root to it and, along the way, for every edge $A_i=v_j$ update the bit vector showing the patterns it may still hit.
This is done efficiently by applying a binary AND operation between the inverted index of $A_i=v_j$ and the current bit vector (initially the bit vector of the root). The number of 1's in the final bit vector after reaching to the leaf node is the number of patterns it hit.
For instance, in Example~\ref{ex:datacollection} (for patterns $P_1$ to $P_6$), assume that still none of the patterns are hit.
Thus, the bit vector associated with the root is $111111$, showing that none of $P_1$ to $P_6$ is so far hit.
Now, consider the value combination $12110$. In the tree data structure, following from the root to the leaf node, we first reach to the edge $A_1=1$. From Figure~\ref{fig:inv1}, the inverted index for this is $111011$.
Hence, after the first AND operation, the patterns it may still hit are: $111111 \land 111011 = 111011$. The next value is $A_2=2$; having the inverse index of this as $111110$, the bit vector gets updated as $111011\land 111110 = 111010$. Following this on the inverse indices of $A_3=1$, $A_4=1$, and $A_5=0$, the final result is $000010$, that is $12110$ only hits $P_5$.}
Using this structure, we design the threshold-based algorithm \greedy \submit{(pseudo-code in\cite{techreport})} that uses the hit-count of its best-known value combination to prune the tree.
The algorithm traverses through the tree data structure in a DFS manner and, starting from the root, it calls the validation oracle before generating each child of a node, to make sure it is semantically meaningful.
As a result, it will output only the value combination that are valid.
The algorithm computes the bit vector of valid children of a node
by applying the binary AND operation between the current bit vector\techrep{ ($filter$)} and the inverted index of each of its children.
The algorithm uses the best-known value combination as a lower-bound threshold to prune the tree.
If the children of the current node are leaf nodes and the best of them hits more patterns than the best-known value combination, the best option gets updated.
For the other nodes, the algorithm prioritizes the children of the current node based on the number of 1's in their bit-vectors. Then, starting from the one with the max count, until the number of 1's in the bit-vectors is more than the best-known hit count, it recursively checks if it can find a better value combination in the subtrees of the children.
\submit{
The algorithm keeps collecting the value combinations that hit the maximum number of remained patterns until all of the patterns in $M_\lambda$ get hit.
}
\techrep{
Following this algorithm for patterns $P_1$ to $P_6$ in Example~\ref{ex:datacollection}, a value combination that hits the maximum number of patterns is 02011, hitting the patterns $P_1, P_3$, and $P_4$.

Now, using this algorithm, we can implement the \greedy approximation algorithm (Algorithm~\ref{alg:wtc}).
The inputs to the algorithm are the set of uncovered patterns at level $\lambda$, i.e., $M_\lambda$, and it returns the set of value combinations to collect. It keeps collecting the value combinations that hit the maximum number of remained patterns until all of the patterns in $M_\lambda$ get hit.
Following the greedy algorithm on Example~\ref{ex:datacollection} on $P_1$ to $P_6$, it suggests collecting three value combinations 02011, 02111, and 10201.

\begin{algorithm}[b]
\footnotesize
\caption{\footnotesize {\bf hit-count} \\
         {\bf Input:} The bit vector $filter$, best-known hit-count $c_{max}$, inverse indices $I$, and the current level $i$\\
         %{\bf Output:} The value combination that hits the max. remaining patterns
        }
\begin{algorithmic}[1]
\label{alg:hitcount}
    \FOR{value $v_j$ in $c_i$}
        \STATE $bv[v_j] = filter \land I_{A_i=v_j}$
        \STATE $cnt[v_j] = $ number of 1's in $bv[v_j]$
    \ENDFOR
    \IF{$i==d$}
        \STATE $v_{max} = \mbox{argmax } cnt[j]$
        \STATE {\bf return} $max(cnt[v_{max}],c_{max})$, $v_{max}$
    \ENDIF
    \STATE sort values $v_j$ in $c_i$ based on $cnt[j]$, descendingly
    \FOR{$j=1$ to $c_i$}
        \IF{$cnt[v_j]<c_{max}$}
            \STATE {\bf break} 
        \ENDIF
        \STATE $tmp_{cnt},tmp = $ {\bf hit-count}($bv[v_j]$, $c_max$, $I$, $i+1$)
        \IF{$tmp_{cnt}>c_{max}$}
            \STATE retval = join($tmp,v_j$)
            \STATE $c_{max}=tmp_{cnt}$
        \ENDIF
    \ENDFOR
    \STATE {\bf return} $c_{max}$, retval
\end{algorithmic}
\end{algorithm}

\begin{algorithm}[b]
\footnotesize
\caption{\footnotesize \greedy \\
         {\bf Input:} The set of uncovered patterns to hit $M$\\
         {\bf Output:} The set of value combinations to collect
        }
\begin{algorithmic}[1]
\label{alg:wtc}
    \STATE $filter=$ a bit vector of size $|M|$ with all bits being $1$
    \STATE $I = $ inverted indices of attribute values to $M$
    \STATE $V = \{\}$
    \WHILE{$\exists ~1\leq j\leq |M|$ s.t. $filter[j]=1$}
        \STATE $c_{max}$, $v$ = {\bf hit-count}($filter$, $0$, $I$, $1$)
        \STATE add $v$ to $V$ and update the $filter$ accordingly
    \ENDWHILE
    \STATE {\bf return} $V$
\end{algorithmic}
\end{algorithm}
}
\revision{
As an implementation note, when the algorithm selects a value combination, it takes the intersection of the patterns it hits to find a more general pattern that any of its matching value combinations hit the same set of patterns. It provides more freedom to the user in the data collection.
}
\section{Experimental Evaluation}\label{sec:exp}
We conducted comprehensive experiments on real data to both validate our proposal and to study the efficiency of the proposed algorithms in practice.

\subsection{Experimental Setup}
\vspace{-2mm}\stitle{Datasets.} Three real datasets were used for the experiments:
\begin{itemize}[leftmargin=*]
\itemsep0em 
\item {\it COMPAS\footnote{\scriptsize \url{www.propublica.org/datastore/dataset/compas-recidivism-risk-score-data-and-analysis}}:} 
ProPublica is a nonprofit organization that produces investigative journalism. They collected and published the COMPAS dataset as part of their investigation into racial bias in criminal risk assessment.
The dataset contains demographics, recidivism scores, and criminal offense information for 6,889 individuals. We used the attributes {\tt sex} (0: male and 1: female), {\tt age} (0: under 20, 1: between 20 and 39, 2: between 40 and 59, and 3: above 60), {\tt race} (0: African-American, 1: Caucasian, 2: Hispanic, and 3: other races), and {\tt marital status} (0: single, 1: married, 2: separated, 3: widowed, 4: significant other, 5: divorced, and 6: unknown) for studying the coverage.\item {\it AirBnB\footnote{\scriptsize \url{www.airbnb.com}}:} 
AirBnB is a popular online peer to peer travel marketplace that provides a framework for people to lease or rent short-term lodging. We use a collection of the information of approximately 2 million \emph{real} properties enlisted in AirBnB.
The website provides 41 attributes for each property, out of which 36 are boolean attributes, such as \texttt{\small TV}, \texttt{\small internet}, \texttt{\small washer}, and \texttt{\small dryer}. %, while 5 are ordinal attributes, such as \emph{Number of Bedrooms} and \emph{Number of Beds}. 
\item {\it BlueNile\footnote{\scriptsize \url{www.bluenile.com/diamond-search?}}:}
Blue Nile is the largest online diamond retailer globally.  We collected its catalog containing 116,300 diamonds at the time of access.
The dataset has 7 categorical attributes for the diamonds, namely {\tt \small shape}, {\tt \small cut}, {\tt \small color}, {\tt \small clarity}, {\tt \small polish}, {\tt \small symmetry}, and {\tt \small florescence} with cardinalities 10, 4, 7, 8, 3, 3, and 5, respectively.
\end{itemize}

\stitle{Hardware and Platform.}
The experiments were conducted on a Linux machine with a 3.8 GHz Intel Xeon processor and 64 GB memory.
The algorithms were implemented in Java.

\stitle{Experiments plan.}
We want to study coverage in real data.  Are there indeed uncovered patterns?
Are these likely to cause errors in prediction or analysis?
%start the experiments by validating our proposal. To do so, we investigate the COMPAS dataset.
We also want to study the performance of the proposed algorithms, both for MUP identification and for coverage enhancement.
We studied both these sets of questions on all three datasets.  
In the interests of space, we report only the most salient results.  In particular, we focus on the COMPAS dataset for the first set of questions, since the negative consequences of lack of coverage are potentially more severe than for other datasets where the impact may be limited to errors in analytical results.  We focus on the AirBnB dataset for the performance questions, since this is the largest of the three datasets.  Since attributes of AirBnB are binary, we supplement with the BlueNile dataset to highlight situations where its much higher cardinality of attribute values matters.

%We use our large datasets, especially AirBnB, for this purpose. Varying the factors that affect the performance of the algorithms, we shall evaluate the algorithms we proposed for Problems 1 and 2.

\subsection{Validation}\label{subsec:exp-validation}
\subsubsection{Issues with Coverage in Real Data} \label{subsec:exp-validation-1}
%We use the COMPAS dataset for validating our proposal.
%Recall that this dataset contains the demographic and criminal offense information of individuals, published by ProPublica.
 
Consider four demographical attributes {\tt \small sex}, {\tt\small age}, {\tt\small race}, and {\tt\small marital status}, as the attributes of interest in the COMPAS dataset.
We investigate the lack of coverage in this dataset with regard to these four attributes to show the risks of using it for important tasks such as
assessing a criminals' likelihood to re-offend and sentencing them accordingly. Setting the threshold to $10$, 
all the single attribute values contain more instances than the threshold. Still, there totally are 65 MUPs in this dataset, out of which $19$ are in level $\ell=2$, 23 in level $\ell=3$, and 23 in level $\ell=4$.
Besides other MUPs, the existence of 19 level two MUPs in the dataset emphasizes the potential of bad predictions for large spaces in the data cube.
To highlight one example, the MUP XX23 shows the lack of coverage for {\it widowed Hispanics}.
The dataset contains only two instances matching this pattern and interestingly both of them have offended multiple times.
In the absence of enough representatives for the minority subgroup,
the trained model, will likely generalize, not sticking to the couple of examples it has seen for the minority subgroup.
However, the generalization becomes problematic when the ``behavior'' in the subgroup is different and the generalization is misleading.
% Still, if the classification rules for the under-represented subgroup are similar to those of the generalization, then under-representation is not a problem at least for machine learning.
This means that the model {\em may} not do a good job in modeling the behaviour of minority sub-groups.
Of course, we use the MUP identification to raise a signal for these lack of coverage cases. Whether or not it is problematic, needs the human expert in the loop.
Lack of coverage in this dataset shows the risk of using it for predicting the behavior of individuals for under-represented groups; it, therefore, questions the decisions made based on such predictions.
To show case an effect of the lack of coverage, next we use this dataset for training a classifier.
\subsubsection{Lack of coverage's effect}
After showing the lack of coverage in the COMPAS dataset, next we conduct an experiment to show its effect on accuracy of a prediction task.
Using the scikit-learn package (version 0.20) on Python, 
we trained a {\em decision tree} as the classifier, while using {\tt \small sex}, {\tt \small age}, {\tt \small race}, and {\tt \small marital status} as the observation attributes.
Using the attribute {\tt \small prior-count}, we created the binary label attribute that shows if a criminal has re-offended.
First, using the cross-validation, we observed that the trained model has acceptable accuracy and f1 measures of 0.76 and 0.7 over a random test set.
Relying on these numbers, a data scientist may consider using this model for predicting the behaviour of criminals.
However, in this experiment we show that these measures does not necessarily show the good performance for the minorities.
We focus on the minority class of Hispanic Females (HF), as there are only 100 of those in the dataset.
We chose this group, specifically because we were limited to the records in the dataset, and were not able to collect additional data points. Hispanic Females 
was (i) a minority subgroup small enough that removing its instances would not noticeably change the size of the training data, while (ii) there were ``enough'' tuples of this group in the dataset (100 tuples) that we could show case the impact of additional data collection.
We considered a randomized set of 20 (out of the 100) HF as the test set for studying the prediction over this group.
Since we not only wanted to study the effect of the lack of coverage, but also the coverage enhancement, we created 5 training sets (using the remaining 80 HF criminals), containing $\{ 0, 20, 40, 60, 80\}$ HF plus all other records not in this demographic.
We used these datasets for training the classifier and calculated its accuracy and f1 measure for predicting our test set of 20 HF.
Figure~\ref{fig:cla} shows the results. The x-axis shows the datasets, while the left and right y-axes show the accuracy and f1 measure for each setting.
Collecting the additional data points should result in more accurate for the under-represented groups, as those are to provide a better understanding of those, while not having a major impact in the overall accuracy. This is confirmed in the figure, as the overall accuracy remained on 76\% in all settings. We also observed that overall f1-measure did not change from 0.7.
First, one can see that the dataset that does not have any HF, has an unacceptable performance for this class, as its accuracy is less than 50\%.
The next observation is that the accuracy and f1 measures improve as the lack of coverage is resolved by adding more HF to the training data. The reduction in the slope of the accuracy curve around 40 suggests that it can be a good choice for the coverage threshold. Interestingly, this is aligned with the central limit theorem's rule of thumb of 30.
In a similar experiment, 
we considered two other minority subgroups, (1. Female - Other Races (FO) and 2. Male - Other Races (MO)) for which there existed at least 20 records in the dataset that we could consider as the test data.
Removing the records of these demographics from the training data, the accuracy of the model was 39\% for FO and 59\% for MO.
The accuracy different between the two groups shows the higher similarity in the ``behaviour'' of MO to other records in the training data.

\begin{figure*}[t]
\vspace{-2mm}
\begin{minipage}{0.23\textwidth}
  	\centering
  	\vspace{6mm}
  	\begin{tikzpicture}
\begin{axis}[
width  = 1.05\linewidth,
height = 4cm,
xtick={0, 1, 2, 3, 4},
xticklabels={0,20,40,60,80},
ylabel = {Accuracy},
ymin = 0.45, ymax=0.8,
ylabel style={
	font=\scriptsize,
    anchor=west,
    at={(0.22,0.25)},
},
ytick pos=left,
xtick align=inside,
xticklabel style={font=\scriptsize},
yticklabel style={font=\scriptsize},
xlabel = {dataset index},
xlabel style={
	font=\scriptsize
},
legend style={
legend columns=1,
at={(1.0,1.0)},
anchor=south east,
column sep=0ex,
draw=none,
font=\scriptsize,
fill=none
}
]
\addplot [mark=square, mark options={fill=white, scale=0.8}, color=rred] table [x=Dataset_index, y=accuracy_overall, col sep=comma] {data/Exp_result.csv};
\addplot [mark=*, mark options={fill=white, scale=0.8}, color=bblue] table [x=Dataset_index, y=accuracy, col sep=comma] {data/Exp_result.csv};
\legend{\revision{Overall Accuracy}, Subgroup Accuracy}
\end{axis}

\begin{axis}[
ybar,
axis y line*=right,
hide x axis,
width  = 1.05\linewidth,
height = 4cm,
hide x axis,
axis y line*=right,
 bar width=0.2cm,
ylabel={f1-measure},
ylabel near ticks,
ylabel style={
	font=\scriptsize,
    at={(1.15,0.5)}
},
yticklabel style={font=\scriptsize},
ytick pos=right
]
\addplot [draw=none,fill = ppurple,fill opacity=0.2] table [x=Dataset_index, y=f1-measure, col sep=comma] {data/Exp_result.csv};
\end{axis}

\end{tikzpicture}
	\vspace{-2mm}\caption{\footnotesize \revision{The effect of lack of coverage on classification: accuracy, f1 measure}}\label{fig:cla}
\end{minipage}
\hfill
\begin{minipage}{0.23\textwidth}
  	\centering
\begin{tikzpicture}
\begin{axis}[
width  = 1.05\linewidth,
height = 4cm,
xmode=log,
xtick={0.000001, 0.00001,  0.0001, 0.001, 0.01},
ymax = 110,
ymin = 0,
ylabel = {Runtime (s)},
ylabel style={
	font=\scriptsize,
    anchor=west,
    at={(0.22,0.25)},
},
xlabel style={
	font=\scriptsize
},
xtick align=inside,
ytick pos=left,
xticklabel style={font=\scriptsize},
yticklabel style={font=\scriptsize},
xlabel = {Threshold rates},
legend style={
legend columns=1,
at={(1.0,1.0)},
anchor=south east,
column sep=0ex,
draw=none,
font=\scriptsize,
fill=none
}
]

\addplot [mark=x, mark options={fill=white, scale=0.8}, color=gray] table [x=threshold, y=apriori, col sep=comma] {data/thresholdTestAirbnb.csv};
\addplot [mark=*, mark options={fill=white, scale=0.8}, color=bblue] table [x=threshold, y=PatternBreakerOriginal, col sep=comma] {data/thresholdTestAirbnb.csv};
\addplot [mark=triangle*, mark options={fill=white, scale=0.8}, color=ggreen] table [x=threshold, y=PatternCombiner, col sep=comma] {data/thresholdTestAirbnb.csv};
\addplot [mark=square*, mark options={fill=white, scale=0.8}, color=rred] table [x=threshold, y=hybrid, col sep=comma] {data/thresholdTestAirbnb.csv};

\legend{\scriptsize \revision{\apriori},\patternbreaker, \patterncombiner, \deepdiver}
\end{axis}

\begin{axis}[
ybar,
axis y line*=right,
hide x axis,
width  = 1.05\linewidth,
height = 4cm,
xmode=log,
ymin = 0,
hide x axis,
axis y line*=right,
 bar width=0.2cm,
ylabel={\# of MUPs},
ylabel near ticks,
ylabel style={
	font=\scriptsize,
    at={(1.05,0.5)}
},
yticklabel style={font=\scriptsize},
ytick pos=right
]

\addplot [draw=none,fill = ppurple,fill opacity=0.2]  table [x=threshold, y=mups, col sep=comma] {data/thresholdTestAirbnb.csv};

\end{axis}

\end{tikzpicture}
\caption{\footnotesize 
\revision{AirBnB: MUP identicication, varying threshold \\($n$ = 1M, $d$ = 15)}
}\label{fig:varyThreshold}  	
\end{minipage}
\hfill
\begin{minipage}{0.23\textwidth}
	\centering
	\begin{tikzpicture}

\begin{axis}[
width  = 1.05\linewidth,
height = 4cm,
xmode=log,
xtick={0.00001,  0.0001, 0.001, 0.01},
ymin = 0,
ylabel = {Runtime (s)},
ylabel style={
	font=\scriptsize,
    anchor=west,
    at={(0.22,0.25)},
},
ytick pos=left,
xtick align=inside,
xticklabel style={font=\scriptsize},
yticklabel style={font=\scriptsize},
xlabel = {Threshold rates},
xlabel style={
	font=\scriptsize
},
legend style={
legend columns=1,
at={(1.0,1.0)},
anchor=south east,
column sep=0ex,
draw=none,
font=\scriptsize,
fill=none
}
]

\addplot [mark=*, mark options={fill=white, scale=0.8}, color=bblue] table [x=threshold, y=PatternBreakerOriginal, col sep=comma] {data/thresholdTestBlueNile.csv};
\addplot [mark=triangle*, mark options={fill=white, scale=0.8}, color=ggreen] table [x=threshold, y=PatternCombiner, col sep=comma] {data/thresholdTestBlueNile.csv};
\addplot [mark=square*, mark options={fill=white, scale=0.8}, color=rred] table [x=threshold, y=hybrid, col sep=comma] {data/thresholdTestBlueNile.csv};

\legend{\patternbreaker, \patterncombiner, \deepdiver}
\end{axis}

\begin{axis}[
ybar,
axis y line*=right,
hide x axis,
width  = 1.05\linewidth,
height = 4cm,
xmode=log,
ymin = 0,
hide x axis,
axis y line*=right,
 bar width=0.2cm,
ylabel={\# of MUPs},
ylabel near ticks,
ylabel style={
	font=\scriptsize,
    at={(1.15,0.5)}
},
yticklabel style={font=\scriptsize},
ytick pos=right
]

\addplot [draw=none,fill = ppurple,fill opacity=0.2]  table [x=threshold, y=mups, col sep=comma] {data/blueNileThresholdMupNum.csv};

\end{axis}

\end{tikzpicture}
\vspace{-6mm}\caption{\footnotesize BlueNile: MUP identicication, varying threshold ($n$ = 116,300, $d$ = 7)}\label{fig:varyThresholdBlueNile}
\end{minipage}
\hfill
\begin{minipage}{0.23\textwidth}
	\centering
	\begin{tikzpicture}

\begin{axis}[
width  = 1.05\linewidth,
height = 4cm,
xmode=log,
ytick={1,10,100},
ymin = 0.9,
ymax = 200,
ymode=log,
ylabel = {Runtime (s)},
xtick align=inside,
xlabel = {Number of data records},
ylabel style={
	font=\scriptsize,
    anchor=west,
    at={(0.22,0.25)},
},
xticklabel style={font=\scriptsize},
yticklabel style={font=\scriptsize},
ytick pos=left,
xlabel style={
	font=\scriptsize
},
legend style={
legend columns=1,
at={(1.0,1.0)},
anchor=south east,
column sep=0ex,
draw=none,
font=\scriptsize,
fill=none
}
]

\addplot [mark=*, mark options={fill=white, scale=0.8}, color=bblue] table [x=n, y=PatternBreakerOriginal, col sep=comma] {data/sizeTestAirbnb.csv};
\addplot [mark=triangle*, mark options={fill=white, scale=0.8}, color=ggreen]  table [x=n, y=PatternCombiner, col sep=comma] {data/sizeTestAirbnb.csv};
\addplot [mark=square*, mark options={fill=white, scale=0.8}, color=rred]  table [x=n, y=hybrid, col sep=comma] {data/sizeTestAirbnb.csv};

\legend{\patternbreaker, \patterncombiner, \deepdiver}

\end{axis}

\begin{axis}[
ybar,
axis y line*=right,
hide x axis,
width  = 1.05\linewidth,
height = 4cm,
xmode=log,
ymin = 0,
hide x axis,
axis y line*=right,
 bar width=0.2cm,
ylabel={\# of MUPs},
ylabel near ticks,
ylabel style={
	font=\scriptsize,
    at={(1.15,0.5)}
},
yticklabel style={font=\scriptsize},
ytick pos=right
]

\addplot [draw=none,fill = ppurple,fill opacity=0.2]  table [x=n, y=mups, col sep=comma] {data/sizeTestAirbnb.csv};

\end{axis}

\end{tikzpicture}
\caption{\footnotesize AirBnB: MUP identicication, varying data size \\($\tau $ = 0.1\%, $d$ = 15)}\label{fig:varySize}
\end{minipage}
\vspace{-2mm}
\end{figure*}
\begin{figure*}[t]
\vspace{-4mm}
\begin{minipage}{0.23\textwidth}
  	\centering
  	\begin{tikzpicture}

\begin{axis}[
width  = 1.05\linewidth,
height = 4cm,
ymode=log,
 xtick={5, 7, 9, 11, 13, 15, 17, 19},
xmax= 18,
ylabel = {Runtime (s)},
ytick align=inside,  
xtick align=inside,
xlabel = {Dimensions},
ytick pos=left,
ylabel style={
	font=\scriptsize,
    anchor=west,
    at={(0.15,0.25)},
},
xticklabel style={font=\scriptsize},
yticklabel style={font=\scriptsize},
xlabel style={
	font=\scriptsize
},
legend style={
legend columns=1,
at={(1.0,1.0)},
anchor=south east,
column sep=0ex,
draw=none,
font=\scriptsize,
fill=none
}
]

\addplot [mark=*, mark options={fill=white, scale=0.5}, color=bblue] table [x=dimension, y=PatternBreakerOriginal, col sep=comma] {data/dimensionTestAirbnb.csv};
\addplot [mark=triangle*, mark options={fill=white, scale=0.5}, color=ggreen]  table [x=dimension, y=PatternCombiner, col sep=comma] {data/dimensionTestAirbnb.csv};
\addplot [mark=square*, mark options={fill=white, scale=0.5}, color=rred]  table [x=dimension, y=hybrid, col sep=comma] {data/dimensionTestAirbnb.csv};

\legend{\patternbreaker, \patterncombiner, \deepdiver}

\end{axis}

\begin{axis}[
ybar,
axis y line*=right,
hide x axis,
ymode=log,
xmax = 18,
width  = 1.05\linewidth,
height = 4cm,
ymin = 0,
hide x axis,
axis y line*=right,
 bar width=0.1cm,
ylabel={\# of MUPs},
ylabel near ticks,
ylabel style={
	font=\scriptsize,
    at={(1.17,0.5)}
},
yticklabel style={font=\scriptsize},
ytick pos=right
]

\addplot [draw=none,fill = ppurple,fill opacity=0.2]  table [x=dimension, y=mups, col sep=comma] {data/dimensionTestAirbnb.csv};
\end{axis}
\end{tikzpicture}
\caption{\footnotesize AirBnB: MUP identicication, varying dimension ($n$ = 1M, $\tau$ =  0.1\%)}\label{fig:varyDimension}
\end{minipage}
\hfill
\begin{minipage}{0.23\textwidth}
  	\centering
    \vspace{3mm}
    \begin{tikzpicture}\hspace{-3mm}
		\begin{axis}[
		width  = 1.05\linewidth,
		height = 4cm,
		ymode=log,
		ymin = 0,
		ylabel = {Runtime (s)},
		ylabel style={
	    font=\scriptsize,
        anchor=west,
        at={(0.15,0.25)},
        },
		xlabel style={
		font=\scriptsize
		},
		ytick align=inside,  
		xtick align=inside,
		xticklabel style={font=\scriptsize},
        yticklabel style={font=\scriptsize},
		% symbolic x coords={level2, level4},
		% xtick=data,
		xlabel = {Dimensions},
		legend style={
		legend columns=2,
		anchor=south east,
		legend style={column sep=1mm},
		draw=none,
		font=\scriptsize,
		fill=none
		}
		]
		\addplot [mark=*, mark options={fill=white, scale=0.8}, color=bblue] table [x=dimension, y=level8, col sep=comma] {data/dimensionLevelLimitAbsolute.csv};
		\addplot [mark=diamond*, mark options={fill=white, scale=0.8}, color=ggreen] table [x=dimension, y=level6, col sep=comma] {data/dimensionLevelLimitAbsolute.csv};
		\addplot [mark=triangle*, mark options={fill=white, scale=0.8}, color=rred] table [x=dimension, y=level4, col sep=comma] {data/dimensionLevelLimitAbsolute.csv};
		\addplot [mark=square*, mark options={fill=white, scale=0.8}, color=ppurple] table [x=dimension, y=level2, col sep=comma] {data/dimensionLevelLimitAbsolute.csv};
		\legend{max $\ell$ = 8, max $\ell$ = 6, max $\ell$ = 4, max $\ell$ = 2}
		\end{axis}
	\end{tikzpicture}
	\vspace{-6mm}\caption{\footnotesize MUPs identification with various dimensions using \deepdiver\ (AirBnB, $n$ = 1M, $\tau$ = 0.1\%)}\label{fig:mupsLevel}
\end{minipage}
\hfill
\begin{minipage}{0.23\textwidth}
	\centering
    \vspace{-2mm}
    \begin{tikzpicture}
		\begin{axis}[
		width  = 1.05\linewidth,
		height = 4cm,
		xmode=log,
		ymode=log,
		ymin = 0,
		ylabel = {Runtime (s)},
		xlabel style={
			font=\scriptsize
		},
		ylabel style={
	    font=\scriptsize,
        anchor=west,
        at={(0.15,0.25)},
        },
        xticklabel style={font=\scriptsize},
        yticklabel style={font=\scriptsize},
		ytick align=inside,  
		xtick align=inside,
		xlabel = {Threshold rates},
		legend style={
			legend columns=2,
			at={(1.2,1.0)},
			anchor=south east,
			column sep=0ex,
			draw=none,
			font=\scriptsize,
			fill=none
		}
		]
	\addplot [mark=triangle*, mark options={fill=white, scale=0.8}, color=blue] table [x=threshold, y=naive, col sep=comma] {data/dcThreshold.csv};
	\addplot [mark=square*, mark options={fill=white, scale=0.8}, color=ppurple] table [x=threshold, y=greedy6, col sep=comma] {data/dcThreshold.csv};
	\addplot [mark=diamond*, mark options={fill=white, scale=0.8}, color=ggreen] table [x=threshold, y=greedy5, col sep=comma] {data/dcThreshold.csv};
	\addplot [mark=*, mark options={fill=white, scale=0.8}, color=rred] table [x=threshold, y=greedy4, col sep=comma] {data/dcThreshold.csv};
	\addplot [mark=otimes*, mark options={fill=white, scale=0.8}, color=bblue] table [x=threshold, y=greedy3, col sep=comma] {data/dcThreshold.csv};
	\legend{Naive ($\ell$=3), Greedy ($\ell$=6), Greedy ($\ell$=5), Greedy ($\ell$=4), Greedy ($\ell$=3)}
	\end{axis}
	\end{tikzpicture}
	\vspace{-6mm}\caption{\footnotesize Coverage Enhancement with various thresholds (AirBnB, $n$ = 1M, $d$ = 13)}\label{fig:dcThreshold1}
\end{minipage}
\hfill
\begin{minipage}{0.23\textwidth}
	\centering
    \vspace{3mm}
	\begin{tikzpicture}
		\begin{axis}[
		width  = 1.05\linewidth,
		height = 4cm,
		ymode=log,
		ymin = 0,
		ylabel = {Runtime (s)},
		ylabel style={
	    font=\scriptsize,
        anchor=west,
        at={(0.15,0.25)},
        },
		xlabel style={
			font=\scriptsize
		},
		ytick align=inside,  
		xtick align=inside,
		xlabel = {Dimensions},
		xticklabel style={font=\scriptsize},
        yticklabel style={font=\scriptsize},
		legend style={
			legend columns=2,
			at={(1.,1.0)},
			anchor=south east,
			legend style={column sep=0.2cm},
			draw=none,
			font=\scriptsize,
			fill=none
		}
		]
	\addplot [mark=*, mark options={fill=white, scale=0.8}, color=bblue] table [x=dimension, y=8, col sep=comma] {data/dataCollectionDimension.csv};
	\addplot [mark=diamond*, mark options={fill=white, scale=0.8}, color=rred] table [x=dimension, y=6, col sep=comma] {data/dataCollectionDimension.csv};
	\addplot [mark=triangle*, mark options={fill=white, scale=0.8}, color=ggreen] table [x=dimension, y=4, col sep=comma] {data/dataCollectionDimension.csv};
	\addplot [mark=square*, mark options={fill=white, scale=0.8}, color=ppurple] table [x=dimension, y=2, col sep=comma] {data/dataCollectionDimension.csv};
	\legend{$\ell$ = 6, $\ell$ = 5, $\ell$ = 4, $\ell$ = 3}
	\end{axis}
	\end{tikzpicture}
	\vspace{-2mm}\caption{\footnotesize Coverage Enhancement with various dimensions using Greedy (AirBnB, $n$ = 1M, $\tau$ = 0.1\%)}\label{fig:dcThreshold}
\end{minipage}\vspace{-8mm}
\end{figure*}
\begin{figure}[t]
  \centering
  
\begin{tikzpicture}

\begin{axis}[
ybar,
width  = \linewidth,
height = 3.5cm,
ymode=log,
 bar width=0.08cm,
xtick={5,10,15,20,25,30,35},
ymin = 0,
ymax = 120000,
ylabel = {input/output size},
ylabel style={
	font=\scriptsize
},
xlabel style={
	font=\scriptsize
},
ytick align=inside,  
xtick align=inside,
xlabel = {Dimensions},
legend style={
legend columns=4,
at={(1.0,1.0)},
anchor=south east,
column sep=0ex,
draw=none,
font=\scriptsize,
fill=none
}
]

% \addplot[draw=red,fill=red!50]
% coordinates {(5,0) (10,0) (15,0) (20,0)}
% \addplot[draw=red,fill=red!50]
% coordinates {(5,0) (10,53) (15,265) (20,1480)};
% \addplot[draw=blue,fill=blue!50]
% coordinates {(5,0) (10,53) (15,265) (20,1480)};

\addplot [fill=red!100!black] table [x=dimension, y=input8, col sep=comma] {data/dataCollectionDimension.csv};
\addplot [fill=blue!100!black] table [x=dimension, y=output8, col sep=comma] {data/dataCollectionDimension.csv};
\addplot [fill=red!80] table [x=dimension, y=input6, col sep=comma] {data/dataCollectionDimension.csv};
\addplot [fill=blue!80] table [x=dimension, y=output6, col sep=comma] {data/dataCollectionDimension.csv};
\addplot [fill=red!50] table [x=dimension, y=input4, col sep=comma] {data/dataCollectionDimension.csv};
\addplot [fill=blue!50] table [x=dimension, y=output4, col sep=comma] {data/dataCollectionDimension.csv};
\addplot [fill=red!20] table [x=dimension, y=input2, col sep=comma] {data/dataCollectionDimension.csv};
\addplot [fill=blue!20] table [x=dimension, y=output2, col sep=comma] {data/dataCollectionDimension.csv};

\legend{input ($\ell$=6), output ($\ell$=6), input ($\ell$=5), output ($\ell$=5),input ($\ell$=4), output ($\ell$=4),input ($\ell$=3), output ($\ell$=3)}
\end{axis}
\end{tikzpicture}
\vspace{-4mm}\caption{\footnotesize Coverage Enhancement with various dimensions using Greedy (AirBnB, $n$ = 1M, $\tau$ = 0.1\%) -- input/output sizes}\label{fig:dcThresholdio}
\vspace{-9mm}
\end{figure}

\revision{
% Next we show case the validation oracle set up, the coverage enhancement, and the generated outputs.
\subsubsection{Coverage Enhancement Quality}
In previous experiment, we showed the quality of coverage enhancement in the sense that it increases the model performance for the under-represented groups, while not impacting the overall performance of the model.
In this experiment, we show
the role of human-in-the-loop by setting up the validation oracle and identifying the MUPs to be covered.
Enforcing the rules of validation oracle while 
expanding the tree data structure used 
by the coverage enhancement algorithm \greedy, the semantic appropriateness (validity) of the output of the coverage enhancement algorithm is guaranteed.
We consider the MUPs discovered in \S~\ref{subsec:exp-validation-1} while targeting to satisfy the coverage level of 2.
In the validation oracle, we rule out (a) the combinations with marital status being unknown and (b) the age group below 20 being not single.
Coverage enhancement suggests to collect \{over 60, other races, widowed\}, \{between 20 and 40, Hispanic, widowed\}, \{over 60, significant other\}, \{other races, divorced\}, and \{other races, widowed\}.
}

\subsection{Performance Evaluation}\label{subsec:expperformance}
%After validating our proposal, in this section, 
%we evaluate the efficiency of the proposed algorithms in \S~\ref{sec:3} and \S~\ref{sec:4}.
We evaluate the performance of (i) the three MUP identification algorithms \patternbreaker, \patterncombiner, and \deepdiver, as well as (ii) the coverage enhancement algorithm.
%for discovering the MUPs in different settings. We also evaluate the performance of the algorithms in \S~\ref{sec:4}.
The Na\"{i}ve algorithm for MUP identification (\S~\ref{subsec:mup-naive}) did not finish for any of the settings within the time limit. Therefore, we did not include it in the results.
For the coverage enhancement problem, we compare the \greedy algorithm (\S~\ref{subsec:HS-efficient}) with the direct implementation of the hitting set's approximation algorithm (\naive).
We use our largest dataset, i.e. AirBnB, as the default and test the algorithms' performances under various coverage threshold ($\tau$) on both AirBnB and BlueNile. We varied the number of attributes ($d$) and the size of the dataset ($n$) on our largest dataset, that is AirBnB.
In addition to the proposed algorithms, for MUP identification, we also consider comparing with \apriori, the following adaptation of apriori algorithm~\cite{apriori}:
we consider each $\langle$attribute,value$\rangle$ as an item and find the frequent item-sets. For each such item-set we find its parents (the item-sets that include the item-set and one more item). For each such parent, if all of its children are frequent, we find the corresponding pattern and add it to the set of MUPs. We understand that
because the items are considered independently in the frequent item-set mining, not all item-set represent a valid pattern. For instance, consider the items $I_1=\langle A_1, 0\rangle$ $I_2=\langle A_1, 1\rangle$. Then the item-set $\{I_1, I_2\}$ does not represent a valid pattern.
Furthermore, this algorithm considers a much larger search space (lattice) to explore, compared to the pattern graph our algorithms explore.
For instance, consider a case where there are 10 attributes, each with cardinality 5. 
The size of the pattern graph is $(5+1)^{10}$, around 60 million nodes, whereas after considering each attribute-value as an item, the size of the defined lattice is $2^{5\times 10}$, around $10^{15}$.

\subsubsection{MUP identification - varying threshold}
%\stitle{Methodology:}
%In this experiment, we study the algorithms' efficiency for the data under different coverage status from mostly uncovered to mostly covered.
%One way to perform such an experiment is by finding different datasets with various coverages. Alternatively,
%Intuitively, we should have actual different datasets with various data coverage to perform this experiment, which means we need to carefully select data records from the raw dataset of millions of data records. Instead, 
For AirBnB, we varied the coverage threshold from 0.0001\% (most patterns are covered) to 1\% (most patterns are uncovered). The dataset size was set to one million, and the number of attributes was set to 15.
For BlueNile, we had 7 attributes and 116,300 records.  We varied coverage from 0.001\% (threshold $=$ 1) to 1\%.
%simulated three datasets with various data coverage by simply using the same dataset and varying the coverage threshold (from 0.0001\% to 1\% for AirBnB, and from 0.0001\% to 1\% for BlueNile). The dataset size was set to one million, and the number of attributes was set to 15 for AirBnB and 7 for BlueNile.

\stitle{Results.}
The runtimes 
%of tests with different thresholds 
are shown in Figure~\ref{fig:varyThreshold} (AirBnb) and Figure~\ref{fig:varyThresholdBlueNile} (BlueNile). The x-axis denotes the different threshold rate values. The left-y-axis is the runtime in seconds, while the right axis and the bars show the output size (number of MUPs).
\revision{
In addition, Figure~\ref{fig:varyThreshold} also contains the results for \apriori, the adaptation of the apriori algorithm for discovering the MUPs.
As explained in \S~\ref{subsec:expperformance}, this algorithm suffers from multiple facts that makes it unsuitable for MUP discovery. First, the lattice data structure it has to explore can be extensively larger than the pattern graph. Second, it needs to generate the parents of the frequent item-sets to find the infrequent item-sets that all of their children are frequent. Finally, not all the discovered item-sets represent valid MUPs.
This is confirmed in this experiment where it only finished for one settings in less than 100 seconds. For instance for threshold of $0.001\%$ it took 516 sec. to finish. As expected, we observed the same behaviour in other experiments as well. 
Hence, in the rest of experiments, we only focus on evaluating the algorithms we proposed in this paper.
}
When the threshold increases, 
%the test dataset becomes less covered. 
larger regions in the space become uncovered and more general MUPs with smaller levels appear in the results.
This is the reason for the drop in the runtime of \patternbreaker\ in Figure~\ref{fig:varyThreshold} and Figure~\ref{fig:varyThresholdBlueNile}.
Recall that \patternbreaker\ is a top-down search algorithm, and generally returns faster when the MUPs are higher in the pattern graph (having small levels).
In contrast,  \patterncombiner's runtime increases as the \patterncombiner\ is a bottom-up search algorithm and, hence, terminates faster when the MUPs are low in the pattern graph (when the space is mostly covered) as shown in Figure~\ref{fig:varyThreshold} and Figure~\ref{fig:varyThresholdBlueNile}.
In tests with AirBnB, these two algorithms have similar speeds when the threshold is around 0.01\%, in which case, most MUPs appear in the middle of the graph. Meanwhile, Figure~\ref{fig:varyThreshold} also shows that \deepdiver\ is as fast, if not faster, as the other two algorithms in all situations. This suggests that the efficiency of \deepdiver\ is more robust to the actual data coverage status. As for BlueNile, Figure~\ref{fig:varyThresholdBlueNile} also suggests \deepdiver\  is the best in all cases, whereas \patterncombiner\ is always slower.
Still the gap between \revision{\patterncombiner\ with} the two other algorithms is larger. The high cardinality of the attributes in BN is the key to this behavior.
In this situation, the width of the pattern graph quickly increases. 
The lowest level (level $7$) of the pattern graph in this case has more than 100K nodes, whereas for 7 binary attributes, it is 128.
Therefore, due to the significant width of the graph in the bottom-level, \patterncombiner\ (the bottom-up algorithm) loses its efficiency.

\subsubsection{MUP identification - varying data size}
%\stitle{Methodology.}
%In this experiment, we evaluate the scalability of the proposed algorithms for different dataset sizes in terms of number of items ($n$).
%Our proposed algorithms should scale up to larger datasets.
Setting the number of attributes to 15 and threshold to 1\%, we evaluated the three MUP identification algorithms on data samples of various sizes from 10K to 1M and measured the runtime. %The numbers of attributes were all 15, and the threshold was set at 1\%. 

\stitle{Results.}
Figure~\ref{fig:varySize} shows the runtime plots.
The x-axis denotes the size of test dataset; the left-y-axis denotes the runtime in seconds and the right-y-axis (and the bars) show the number of MUPs. 
All three algorithms had running time only slightly impacted by data set size, taking less than 100 seconds in all settings.  The effort is driven more by the number of patterns, which is independent of data set size.
The \patterncombiner\ algorithm checks the actual dataset only for the bottom layer of the pattern graph and %pass these number as intermediate results to other layers, the increase of the cost of checking the datasetdoes not have a major impact on its performance. %should barely impact the runtime of \patterncombiner\ which can be verified in the figure. 
so the data set size has no effect on most of its computation.
\patternbreaker\ and \deepdiver\ need to check the data for computing the coverage of the intermediate nodes, so data set size does matter.  However, the use of inverted indices limits the impact.
%Still, using the inverted indices for computing the coverage (\S~\ref{ap:inverted}), the impact of increasing $n$ to the running time of these algorithms is not major.
%The runtimes of the other algorithms, \patternbreaker, and \deepdiver\ do grow faster than that of \patterncombiner\ as the size of the test dataset increases. However, such growths are sublinear to the growth of the data size. Therefore, these results provide strong evidence that the three algorithms are scalable when data size increases.  

\subsubsection{MUP identification - varying data dimensions}
%\stitle{Methodology.}
Similarly, we evaluate the scalability of the proposed algorithm as the number of attributes ($d$) increases.
%the proposed algorithms should scale up to high dimensional data.
With a dataset size of one million records and the threshold set at 1\%, we measured the overall runtime of all three algorithms with the dataset projected down to between 5 and 17 dimensions.

\stitle{Results.}
In Figure~\ref{fig:varyDimension}, the x-axis denotes the number of attributes, while the left-y-axis and right-y-axis (the bars) denote the runtime in seconds and the output size, respectively.
%As explained in \S~\ref{subsec:patterngraph}, 
The size of the pattern graph increases exponentially with the number of attributes. 
%Also, recall that, as proved in Theorem~\ref{th:1}, there exists no polynomial-time algorithm that can enumerate the MUPs as the output size itself can increase exponentially as the input size increases.
%This is confirmed in the figure as the number of MUPs increase from tens to millions when the number of attributes increases from 5 to 17.
%As the number of attributes increases, the search space increases exponentially, and the MUP identification algorithms have to spend more time exploring it. %is exponential to the data dimensionality, the challenge of searching for MUPs increases significantly. 
%Hence, in Figure~\ref{fig:varyDimension}, we observe that, as expected, the running time of all three algorithms increase accordingly. %as the number of dimensions increases. 
The number of MUPs and the algorithm running times also increase exponentially.
Still, all algorithms managed to finish in a reasonable time (under two minutes) for up to 17 attributes.
%However, \patternbreaker\ and \deepdiver\ grow slower than \patterncombiner, which provide good evidence that \patternbreaker\ and \deepdiver\ scale up better than \patterncombiner\ with high dimensional data.

%Even though, as the number of attributes have a major effect of the algorithms' runtime, the 

As the number of attributes increases, the number of MUPs increases exponentially, but those become the combination of more attributes.
While the MUPs with fewer are harmful and important to discover, the MUPs with more attributes are too specific, and hence, less interesting. For example, while lack of coverage for Hispanic males in a dataset is an important fact to discover, not having enough married Hispanic males under the age of 20 is less harmful.
Limiting the exploration level to a certain number, allows the MUP identification algorithms to scale for datasets with tens of attributes and still finding the risky MUPs.
We evaluated this by limiting the MUP discovery level in Figure~\ref{fig:mupsLevel} while using \deepdiver\ for the identification.
As observed in the figure, the algorithm was able to quickly find MUPs of up to level 2 (the MUPs that are the combinations of one or two attributes) for even 35 attributes in around 10 sec.
%Similarly, it could find the MUPs of up to level 4 (the combinations of up to four attributes) in a few minutes for the largest setting.

\subsubsection{Coverage enhancement - varying threshold}
%Next, we evaluate our proposal for the coverage enhancement problem (\S~\ref{sec:4}).
Recall that the objective is to identify the minimum additional data to collect, such that after the data collection the maximum coverage level is not less than $\lambda$, i.e. there are no uncovered patterns on or above a given level $\lambda$.
%We use the AirBnB dataset for the coverage enhancement experiments.
Setting the number of items to 1M in the AirBnB dataset and number of attributes to 13, we vary the threshold rate from $10^{-6}$ to $0.01$ while choosing different maximum coverage levels from $3$ to $6$.

\stitle{Results.}
Figure~\ref{fig:dcThreshold1} represents the experiment results.
%for varying the threshold. 
The x-axis shows threshold and the y-axis provides the runtime in seconds.
First, the single blue triangular tick mark in the top-left of the plot shows the only setting for which the \naive\ algorithm finished within the time limit.
\greedy, on the other hand, finished in a few seconds for all settings.
The next observation is that, as expected, the runtime of the \greedy\ algorithm increases by the level; that is because it needs to collect more data points to ensure that there is no uncovered pattern on or above level $\lambda$, i.e., $\forall P\in\mathcal{M}:~\ell(P)\geq \lambda$.
Also, as the threshold rate increases the MUPs move to the top of the pattern graph. Therefore, more regions in the space become uncovered and more data points are required to guarantee the given maximum coverage level. As a result, the algorithm's runtime increases by the threshold.

\subsubsection{Coverage enhancement - varying data dimensions}
Lastly, we study the effect of the number of attributes on the performance of \greedy, as well as input and output sizes.
Using the AirBnB dataset, while setting the number of items to 1M and the threshold to \%1, we vary the number of attributes from $5$ to $35$, and the max. coverage level from $3$ to $6$.

\stitle{Results.}
Figure~\ref{fig:dcThreshold} shows the runtime of the algorithm, while Figure~\ref{fig:dcThresholdio} provides information about the input and output sizes. Here, by the input size, we refer to the number of uncovered patterns (to cover) at the given level $\lambda$ while the output size is the number of additional data points to collect.
First, as explained above, increasing the maximum coverage level increases the runtime of the algorithm, as the output size increases. This is also reflected in Figure~\ref{fig:dcThresholdio}, as for a fixed number of attributes, both the input and output size increase in orders of magnitude.
Similarly, increasing the number of attributes increases the size of the pattern graph exponentially, and also does the algorithm runtime (Figure~\ref{fig:dcThreshold}) and the output size (Figure~\ref{fig:dcThresholdio}).
Still, recall that lack of coverage for the patterns that are the combination of a few attribute values (having smaller levels) is more harmful than the ones in the form of the combination of several attribute values.
Looking at Figure~\ref{fig:dcThreshold}, while solving the coverage enhancement problem for larger levels takes more time, the algorithm has a reasonable performance for resolving the lack of coverage for smaller values of maximum coverage level.
Finally, in Figure~\ref{fig:dcThresholdio}, applying the greedy approximation algorithm, the output sizes are significantly smaller than the input sizes for each setting. That is because every value combination in the output hits multiple uncovered patterns in the input.
%In Figures~\ref{fig:dcThreshold} and~\ref{fig:dcThresholdio} there are no information for level$=6$ and $d=5$. That simply is because the maximum level in the pattern graph is $5$.

\section{Related Work}\label{sec:related}
Diversity, as a general term for capturing the quality of a collection of items on the variety of its constituent elements~\cite{diversity-jag}, is an important issue in a wide range of contexts, including social science~\cite{simpson1949measurement}, political science~\cite{surowiecki2005wisdom}, information retrieval~\cite{agrawal2009diversifying}, and big data environments and data ethics~\cite{barocas2016big, diversity-jag}.
Facility dispersion problems~\cite{facility1} tend to disperse a set of points such that the minimum or average distance between the pair of points is maximized.
Also, techniques such as determinantal point process (DPP) have been used for diverse sampling~\cite{dpp1,dpp2}.
A recent work~\cite{celis2016fair} considers diversity as the entropy over one discrete low-cardinality attribute. Our definition of coverage can be seen
as a generalization of this, defined over combinations of
multiple attributes.

The rich body of work on sampling, especially in the database community, aims to draw samples from a large database~\cite{DBsampling1,DBsampling2}.
Our goal in this paper is to ensure that a given dataset (often called as ``found data'') is appropriate for use in a data science task.
The dataset could be collected independently, through a process on which the data scientist have limited, or no, control.
This is different from sampling.  

Technically speaking, there are similarities between the algorithms provided in this paper and 
the classical powerset lattice and combinatorial set enumeration problems~\cite{setenum},
such as data cube modeling~\cite{harinarayan1996implementing}, frequent item-sets and association rule mining~\cite{apriori}, data profiling~\cite{heise2013scalable},\techrep{ recommendation systems~\cite{asudeh2017assisting},} and data cleaning~\cite{he2016interactive}.
\revision{
While such work, and the algorithms such as apriori, traverse over the powerset lattice, our problem is modeled as the traversal over the pattern graph which has a different structure (and properties) compared to a powerset lattice.
Hence, those techniques cannot be directly applied here.
We provided some rules for traversing the pattern graph that are inspired from the set enumeration tree~\cite{setenum}, one-to-all broadcast in a hypercube~\cite{bertsekas1991optimal}, and lattice traversal heuristic proposed in~\cite{heise2013scalable}.
}
In \S~\ref{sec:4}, we modeled the data collection problem as a hitting set instance (an equivalent of the set cover problem). Further details about this fundamental problem can be found in references such as~\cite{vazirani2013approximation,hochbaum1996approximation}.

\section{Final Remarks}\label{sec:conclusion}
In this paper, we studied lack of coverage as a risk to using a dataset for analysis.
Lack of coverage in the dataset may cause errors in outcomes, including algorithmic racism.
Defining the coverage over multiple categorical attributes, we developed techniques for identifying the spots not properly covered by data to help the dataset users; we also proposed techniques to help the dataset owners resolve the coverage issues by additional data collection.
Comprehensive experiments over real datasets demonstrated the validity of our proposal.

\revision{
Following ideas such as~\cite{sudman1976applied},
in MUP identification problem, we considered a fixed threshold across different value combinations, representing ``minor subgroups''.
We consider further investigations on identifying threshold value and minor subgroups, as well as other alternatives for future work.
}

\section{Acknowledgements}
\noindent
This work was supported by NSF Grant No. 1741022.
We are grateful to the University of Toronto, Department of Computer Science, and Dr. Nick Koudas for the AirBnB dataset.

% \pagebreak
% \vfill\null
\bibliographystyle{unsrt}
\bibliography{ref}

\begin{thebibliography}{10}

\bibitem{google-gorilla}
M.~Mulshine.
\newblock A major flaw in google's algorithm allegedly tagged two black
  people's faces with the word 'gorillas'.
\newblock Business Insider, 2015.

\bibitem{closed-eyes}
Adam Rose.
\newblock Are face-detection cameras racist?
\newblock Time Business, 2010.

\bibitem{hp1}
Mallory Simon.
\newblock {HP} looking into claim webcams can't see black people.
\newblock CNN, 2009.

\bibitem{hp2}
Tess Townsend.
\newblock Most engineers are white and so are the faces they use to train
  software.
\newblock Recode, 2017.

\bibitem{propublica}
Julia Angwin, Jeff Larson, Surya Mattu, and Lauren Kirchner.
\newblock Machine bias: Risk assessments in criminal sentencing.
\newblock ProPublica, 5/23/2016.

\bibitem{google-gorilla-resolution}
Alex Hern.
\newblock Google's solution to accidental algorithmic racism: ban gorillas.
\newblock The Guardian, 2018.

\bibitem{chen2018my}
Irene Chen, Fredrik~D Johansson, and David Sontag.
\newblock Why is my classifier discriminatory?
\newblock In {\em NeurIPS}, 2018.

\bibitem{diversity-jag}
Marina Drosou, HV~Jagadish, Evaggelia Pitoura, and Julia Stoyanovich.
\newblock Diversity in big data: A review.
\newblock {\em Big data}, 5(2), 2017.

\bibitem{biggio2013evasion}
Battista Biggio, Igino Corona, Davide Maiorca, Blaine Nelson, Nedim
  {\v{S}}rndi{\'c}, Pavel Laskov, Giorgio Giacinto, and Fabio Roli.
\newblock Evasion attacks against machine learning at test time.
\newblock In {\em ECML PKDD}, 2013.

\bibitem{yang2018nutritional}
Ke~Yang, Julia Stoyanovich, Abolfazl Asudeh, Bill Howe, HV~Jagadish, and Gerome
  Miklau.
\newblock A nutritional label for rankings.
\newblock In {\em SIGMOD}, 2018.

\bibitem{setenum}
Ron Rymon.
\newblock Search through systematic set enumeration.
\newblock Technical report, University of Pennsylvania, 1992.

\bibitem{apriori}
Rakesh Agrawal and Ramakrishnan Srikant.
\newblock Fast algorithms for mining association rules.
\newblock In {\em VLDB}, 1994.

\bibitem{vazirani2013approximation}
Vijay~V Vazirani.
\newblock {\em Approximation algorithms}.
\newblock Springer Science \& Business Media, 2013.

\bibitem{sudman1976applied}
Seymour Sudman.
\newblock Applied sampling.
\newblock {\em Academic Press New York}, 1976.

\bibitem{invertedindex1}
Doug Cutting and Jan Pedersen.
\newblock Optimization for dynamic inverted index maintenance.
\newblock In {\em SIGIR}, 1989.

\bibitem{simpson1949measurement}
Edward~H Simpson.
\newblock Measurement of diversity.
\newblock {\em Nature}, 163(4148), 1949.

\bibitem{surowiecki2005wisdom}
James Surowiecki.
\newblock {\em The wisdom of crowds}.
\newblock Anchor, 2005.

\bibitem{agrawal2009diversifying}
Rakesh Agrawal, Sreenivas Gollapudi, Alan Halverson, and Samuel Ieong.
\newblock Diversifying search results.
\newblock In {\em WSDM}, pages 5--14. ACM, 2009.

\bibitem{barocas2016big}
Solon Barocas and Andrew~D Selbst.
\newblock Big data's disparate impact.
\newblock {\em Cal. L. Rev.}, 104:671, 2016.

\bibitem{facility1}
SS~Ravi, Daniel~J Rosenkrantz, and Giri~Kumar Tayi.
\newblock Facility dispersion problems: Heuristics and special cases.
\newblock In {\em WADS}. Springer, 1991.

\bibitem{dpp1}
Alex Kulesza, Ben Taskar, et~al.
\newblock Determinantal point processes for machine learning.
\newblock {\em Foundations and Trends in ML}, 5(2--3), 2012.

\bibitem{dpp2}
N.~Anari, Sh.~O. Gharan, and A.~Rezaei.
\newblock Monte carlo markov chain algorithms for sampling strongly rayleigh
  distributions and determinantal point processes.
\newblock In {\em COLT}, pages 103--115, 2016.

\bibitem{celis2016fair}
L~Elisa Celis, Amit Deshpande, Tarun Kathuria, and Nisheeth~K Vishnoi.
\newblock How to be fair and diverse?
\newblock {\em CoRR}, abs/1610.07183, 2016.

\bibitem{DBsampling1}
Frank Olken and Doron Rotem.
\newblock Random sampling from databases: a survey.
\newblock {\em Statistics and Computing}, 5(1):25--42, 1995.

\bibitem{DBsampling2}
Graham Cormode, Minos Garofalakis, Peter~J Haas, Chris Jermaine, et~al.
\newblock Synopses for massive data: Samples, histograms, wavelets, sketches.
\newblock {\em Foundations and Trends in Databases}, 4(1--3):1--294, 2011.

\bibitem{harinarayan1996implementing}
Venky Harinarayan, Anand Rajaraman, and Jeffrey~D Ullman.
\newblock Implementing data cubes efficiently.
\newblock In {\em SIGMOD}, 1996.

\bibitem{heise2013scalable}
A.~Heise, J.A. Quian{\'e}-Ruiz, Z.~Abedjan, A.~Jentzsch, and F.~Naumann.
\newblock Scalable discovery of unique column combinations.
\newblock {\em PVLDB}, 2013.

\bibitem{asudeh2017assisting}
Abolfazl Asudeh, Azade Nazi, Nick Koudas, and Gautam Das.
\newblock Assisting service providers in peer-to-peer marketplaces: Maximizing
  gain over flexible attributes.
\newblock {\em CoRR}, abs/1705.03028, 2017.

\bibitem{he2016interactive}
J.~He, E.~Veltri, D.~Santoro, G.~Li, G.~Mecca, P.~Papotti, and N.~Tang.
\newblock Interactive and deterministic data cleaning.
\newblock In {\em SIGMOD}, 2016.

\bibitem{bertsekas1991optimal}
D.P. Bertsekas, C.~{\"O}zveren, G.D. Stamoulis, P.~Tseng, and J.N. Tsitsiklis.
\newblock Optimal communication algorithms for hypercubes.
\newblock {\em JPDC}, 11(4), 1991.

\bibitem{hochbaum1996approximation}
Dorit~S Hochbaum.
\newblock {\em Approximation algorithms for NP-hard problems}.
\newblock PWS Publishing Co., 1996.

\end{thebibliography}
%  \vspace{10cm}
 
\techrep{
\appendix
% \vspace{5mm}
% {\bf \LARGE APPENDIX}
\subsection{Coverage Computation}\label{ap:inverted}
Computing the coverage of a pattern is a key operation for identifying the MUPs.
Therefore, in this subsection, we design the coverage oracle $cov$, that given a pattern $P$ (as well as the dataset $\mathcal{D}$), returns $cov(P)$.

The direct implementation of the oracle passes through the dataset once, and follows Definition~\ref{def:coverage}, literally.
Instead, we use inverted indices~\cite{invertedindex1}.
For every value $v_j$ for an attribute $A_i$, we consider a bit vector $I_{i,j}$. Also, we aggregate the items with the same value combinations.
Let $D$ be the set of unique value combinations in $\mathcal{D}$ while for each entry in $D$ we maintain the number of items in $\mathcal{D}$ with that value combination.
The $k$-th bit of a vector $I_{i,j}$ is $1$, if $D_k[i] = v_j$ (and is $0$ otherwise).
In order to compute the coverage of a pattern $P$, it applies the binary AND operations between the corresponding vectors for the deterministic elements in $P$ and takes the dot product of the result vector with the count vector that shows the number of items matching each value combination.
For instance, consider Example~\ref{ex:1}.
The following are the bit vectors and the count vector for this example:
\begin{small}
	\begin{center}
	\begin{tabular}{|@{}c@{}|@{}c@{}|@{}c@{}|@{}c@{}|@{}c@{}||@{}c@{}|@{}c@{}|@{}c@{}|@{}c@{}|@{}c@{}||@{}c@{}|@{}c@{}|@{}c@{}|@{}c@{}|@{}c@{}|}
		\hline
		vid & $000$&$001$&$010$&$011$&vid & $000$&$001$&$010$&$011$ &vid & $000$&$001$&$010$&$011$ \\ \hline\hline
        $v_{1,0}$&1&1&1&1 &$v_{2,0}$&1&1&0&0 & $v_{3,0}$&1&0&1&0\\ \hline
        $v_{1,1}$&0&0&0&0 &$v_{2,1}$&0&0&1&1 & $v_{3,1}$&0&1&0&1\\ \hline
        {\it cnt}    &1&2&1&1  \\ \cline{1-5}
	\end{tabular}
    \end{center}
\end{small}
Then the coverage of the pattern $0X1$, for example, is calculated by applying the binary AND operation on the vectors $v_{1,1}$ and $v_{3,0}$ and taking the dot product of the result with the {\it cnt} vector. As a result, $cov(0X1)=3$.

The total number of bit vectors are $cd$, as one vector is maintained for each attribute value. Since the length of each bit vector is in $O(n)$, the storage  requirement for the bit vectors is $O(cdn)$.

\subsection{Efficient Dominance Checking}\label{ap:dominancechecking}
Due to the large amount of node visits, efficient MUP dominance checking is critical.
A \naive\ design of this operation is a simple explorative search among all MUPs; this, however, can be expensive given the potentially large number of MUPs.

Instead, similar to Appendix~\ref{ap:inverted}, we use inverted indices.
We create an inverted index for each value of each attribute $A_i$. For each attribute we also consider an inverted index for the patterns that have non-deterministic elements on that attribute.
For each attribute value, we consider a bit vector of size of the current set of discovered MUPs.
When a new MUP is identified, a new bit of 1 is added to the bit vectors corresponding to the values of the elements of the new MUP;
the rest of bit vectors are appended by 0.

To check if a pattern $P$ dominates the current set of MUPs $\mathcal{M}$, we iterate through each value of pattern $P$ and skip the non-deterministic elements as the patterns dominated by $P$ can have any value on those.
Applying the binary AND operation between the bit vectors for the values of the deterministic elements identify if $P$ dominates $\mathcal{M}$: if there is a non-zero element in the result vector, there exists a pattern $P'$ that is dominated by $P$ and, therefore, $P$ dominates $\mathcal{M}$.

Similarly, to check if a pattern $P$ is dominated by the current MUPs $\mathcal{M}$, we parse each value of $P$ but collect a different set of bit vectors:
(i) for the non-deterministic elements ($X$ values) in $P$, we collect the corresponding bit vectors for the non-deterministic elements and (ii) for the deterministic elements in $P$, we collect the result of bitwise OR operation between the bit vector for this attribute value and the bit vector for value $X$ in this attribute.
In the end, we perform the binary AND operation over the collected vectors: if the result is all zero, $\mathcal{M}$ does not dominate $P$; otherwise, $\mathcal{M}$ dominates $P$.

In addition, since we are only interested to know if there is a 1 in the result of the AND operations, 
we apply an early stop strategy by conducting the operation word by word and terminating it as soon as a 1 is observed in the results.
\subsection{Uncovered patterns to hit for maximum coverage level assurance}\label{ap:ptc}
Here we discuss the set of patterns we need to hit (c.f. \S~\ref{sec:4}) in order to guarantee the maximum coverage level for a user-specified value $\lambda$.
In the first glance, applying the hitting set on the MUPs with level at most $\lambda$ will cover them all and, thus, assures the maximum coverage level of $\lambda$.
But, this is not correct.
To further explain it, consider again Example~\ref{ex:datacollection} and the collection of its MUPs $P_1$ to $P_7$.
One may notice that the value combinations 02011, 02111, and 10201 (discovered by the greedy hitting set for $P_1$ to $P_6$) also cover the pattern $P_7:$ X020X, which means all of the MUPs in the example are covered.
Consider the pattern $P:$ 1X11X in level $\ell=3$. First, this is uncovered, as it is a child of the MUP $P_5:$ XX11X.
None of the combinations 02011, 02111, or 10201 match $P$. This means that there exists at least one pattern at level $\ell=3$ that remains uncovered.
This contradicts the claim that the maximum coverage level is $\ell=3$.

In fact, every MUP represents the set of uncovered patterns at lower levels that are connected to it. Covering the MUP may not satisfy covering those patterns as well.
As a result, even though the MUP itself is covered, some of its children may still not be.

On the other hand, if all the (not necessarily maximal) uncovered patterns at a level $\lambda$ are covered, all the more general patterns with level $\leq \lambda$ are also covered. That is because every (more general) pattern $P$ at a level less than $\lambda$
represents a set of value combinations that are the superset of the matches for at least one uncovered pattern $P'$ at level $\ell$ (for the ease of explanation, we say $P'$ is a ``subset'' pattern for $P$). Thus, collecting a value combination that matches $P'$ will also match $P$.

As a result, in order to guarantee the coverage at level $\lambda$, it is enough to apply the greedy hitting set as explained in \S~\ref{sec:4} on the set of uncovered patterns at level $\lambda$.
But, we first need to generate all uncovered patterns at level $\lambda$.
To do so, for every MUP $P\in\mathcal{M}$ where $\ell(P)\leq \lambda$, we find its descendants at level $\lambda$ by replacing $(\lambda - \ell)$ non-deterministic elements (X's) with deterministic values.
For example, in Example~\ref{ex:datacollection}, the subset patterns for $P_1:$ XX01X at level $\ell=3$ are: 0X01X, 1X01X, X001X, X101X, X201X, XX010, and XX011.

}

\end{document}